\documentclass[11pt]{article}
\usepackage[utf8]{inputenc}
\usepackage{fullpage}
\usepackage{svg}
\usepackage{bm}

\usepackage{caption}

\usepackage{booktabs}

\usepackage{algorithm}
\usepackage{algpseudocode}

\usepackage{wrapfig}

\usepackage{amsmath, amssymb, amsthm}
\usepackage[hidelinks]{hyperref}
\usepackage{thmtools}
\usepackage{nameref,cleveref}
\usepackage{color}
\usepackage[shortlabels]{enumitem}

\usepackage{svg}

\def\Comments{1}

\ifnum\Comments>0
\newcommand{\pasin}[1]{\textcolor{red}{{\bf (Pasin:} {#1}{\bf ) }} \marginpar{\tiny\bf
             \begin{minipage}[t]{0.5in}
               \raggedright S:
            \end{minipage}}}

\newcommand{\charlie}[1]{\textcolor{red}{{\bf (Charlie:} {#1}{\bf ) }} \marginpar{\tiny\bf
             \begin{minipage}[t]{0.5in}
               \raggedright S:
            \end{minipage}}}
\else
\newcommand{\pasin}[1]{}
\newcommand{\charlie}[1]{}
\fi

\newtheorem{thm}{Theorem}
\newtheorem{proposition}[thm]{Proposition}

\newtheorem{lemma}[thm]{Lemma}
\newtheorem{theorem}[thm]{Theorem}

\newtheorem{definition}[thm]{Definition}

\newtheorem{remark}[thm]{Remark}
\crefname{definition}{definition}{definitions}
\Crefname{definition}{Definition}{Definitions}

\newcommand{\eps}{\varepsilon}

\newcommand{\R}{\mathbb{R}}
\newcommand{\Z}{\mathbb{Z}}
\newcommand{\N}{\mathbb{N}}
\newcommand{\bbP}{\mathbb{P}}
\newcommand{\E}{\mathbb{E}}

\newcommand{\dr}[2]{\mathrm{D}_{\alpha}\left(#1\ \middle\|\ #2\right)}
\newcommand{\drp}[2]{\mathrm{D}_{\alpha}^{\delta}\left(#1\ \middle\|\ #2\right)}
\newcommand{\dra}[3]{\mathrm{D}_{#3}\left(#1\ \middle\|\ #2\right)}
\newcommand{\drpd}[3]{\mathrm{D}_{\alpha}^{#1}\left(#2\ \middle\|\ #3\right)}

\newcommand{\tP}{\tilde{P}}
\newcommand{\tQ}{\tilde{Q}}
\newcommand{\teps}{\tilde{\eps}}
\newcommand{\hs}[3]{\mathrm{D}^{\mathrm{hs}}_{e^{#3}}\left(#1\ \middle\|\ #2\right)}
\newcommand{\Brac}[1]{\left[#1\right]}
\newcommand{\Paren}[1]{\left(#1\right)}
\newcommand{\disc}{\Delta^{\mathrm{disc}}}
\newcommand{\ndisc}{N_{\mathrm{disc}}}
\newcommand{\discprim}{\psi}

\allowdisplaybreaks

\DeclareMathOperator{\Ber}{Ber}

\newcommand{\piopt}{\pi^*}

\title{Optimal partition selection with R{\'e}nyi differential privacy}

\author{Charlie Harrison \\Google\\ \texttt{\small csharrison@google.com}
\and
Pasin Manurangsi \\Google Research\\ \texttt{\small pasin@google.com}
} 

\begin{document}

\maketitle

\begin{abstract}
A common problem in private data analysis is the partition selection problem, where each user holds a set of partitions (e.g. keys in a GROUP BY operation) from a possibly unbounded set. The challenge here is in maximizing the set of released partitions while respecting a differential privacy constraint. Previous work \cite{desfontaines-partition} presented an optimal $(\eps, \delta)$-DP algorithm when each user submits only a single partition. We generalize this approach to find the optimal algorithm under $\delta$-approximate $(\alpha, \eps)$-R{\'e}nyi differential privacy (RDP), which allows much tighter analysis under composition. Motivated by the \emph{non-existence} of a general optimality result in the case where users submit multiple partitions each, we present an extension of our optimal algorithm tuned for $L^2$ bounded weighted partition selection which can be used as a drop-in improvement over the Gaussian mechanism any time the partition frequency is not also needed. We show that our primitive can be easily plugged into state of the art partition selection algorithms (PolicyGaussian from \cite{gopi20a} and MAD2R from \cite{chen2025scalable}), improving performance both for parallel and sequential adaptive algorithms.
Finally, we show that there is an inherent cost to algorithms which \emph{do} support releasing the frequency as well as the partitions. Specifically, we formulate a basic notion of optimal approximate RDP algorithm for partition selection using additive noise, and show that there is a numerical separation between additive and non-additive noise mechanisms for this problem.
\end{abstract}

\section{Introduction}

Differential privacy \cite{dworkcalibrating} is a strong notion of privacy which bounds the information a worst-case attacker can learn about the output of a privacy mechanism. Since its inception, various relaxations of differential privacy have gained in popularity. In particular, R{\'e}nyi differential privacy (RDP) \cite{Mironov17} and approximate RDP \cite{papernot2022hyperparameter} are most relevant to this work. These relaxations typically provide much better utility in scenarios where many mechanisms are composed together.

The problem of identifying the set of ``keys" in a dataset is common in private data analysis. This comes up for instance, when determining the set of output partitions in a private GROUP BY query, or when trying to release a dataset of high dimensional items like arbitrary strings or URLs. In many cases, the set of possible outputs is exponential or even infinite, which precludes (non-trivial) mechanisms which introduce false positives. Therefore, we limit our approach to cases where the mechanism must emit only a subset of the true partitions.

\subsection{Related work}

Some of the earliest work in private partition selection comes from \cite{Korolova-search-queries}, who analyzed search logs to privately release the set of user queries. Their approach leveraged adding additive Laplace noise to the count of each query, and only keeping queries whose count exceeded some threshold. This basic recipe (sometimes replacing Laplace noise with Gaussian noise) became a popular approach to partition selection, and has been proposed in various private query engines \cite{wilson2020differentially, plume2022}.

In the case where a user only contributes to a single partition, \cite{desfontaines-partition} proved the \emph{optimal} algorithm under $(\eps, \delta)$-DP. That is, each partition is released with the maximum possible probability under the privacy constraint. Our work generalizes this result, as taking the limit $\alpha \to \infty$ under $\delta$-approximate $(\alpha, \eps)$-RDP recovers $(\eps, \delta)$-DP exactly. By allowing a relaxation to finite $\alpha$, we can show utility improvements when under composition (or, equivalently, to handle users contributing multiple partitions).

\cite{gopi20a} formalized the partition selection problem as ``differentially private set union", and innovated on algorithms to handle cases where users contribute more than one partitions. Their proposed algorithm is adaptive and greedy, where each user's contribution is dependent on previous users'. \cite{Carvalho2022IncorporatingIF} follow a similar recipe, proceeding user-by-user in a dependent manner. While these sequential algorithms typically allow releasing the most partitions for the same privacy budget, they are difficult to scale for large datasets or query systems. \cite{swanberg2023dp, chen2025scalable} both attempt to design adaptive algorithms that can be implemented in a highly parallelizable fashion.

Finally, our work involves numerically computing optimal additive noise mechanisms. This overall approach was used in \cite{gilani2025optimal} to find variance-optimal mechanisms under an RDP constraint, with the goal of improving an overall privacy guarantee under composition. Unlike their approach which attempted to minimize variance, we minimize \emph{tail bounds}, enabling us to choose a tighter truncation threshold and therefore release more partitions.

\subsection{Our contributions}

\paragraph{Optimal partition selection (\Cref{sec:main-result}).} In \Cref{thm:opt-partition-prim}, we present an optimal partition selection algorithm (under approximate RDP) when users only submit a single partition. Our algorithm exactly recovers the result from \cite{desfontaines-partition} in the limit when $\alpha \to \infty$, but for finite $\alpha$ we can leverage the tight composition afforded by RDP. We additionally show that when users can submit more than one partition, there is no single optimal mechanism (\Cref{thm:no-optimal-multi}).

\paragraph{Weighted partition selection with the SNAPS mechanism (\Cref{sec:weighted-selection}).} We extend our core algorithm to handle real valued weights associated with partitions. Under this framework, we derive the SNAPS algorithm (Smooth Norm-Aware Partition Selection) which satisfy the privacy constraints when the users hold vectors with sensitivities bounded by arbitrary $L^r$ norms. In particular, when $r = 2$, we can plug our algorithm into more complex adaptive mechanisms like DP-SIPS \cite{swanberg2023dp}, MAD2R \cite{chen2025scalable}, or PolicyGaussian \cite{gopi20a} which use the Gaussian mechanism as a subroutine for partition selection.

\paragraph{Numerical experiments (\Cref{sec:exp}).} 
We replace the Gaussian mechanism with the SNAPS mechanism described above in two weighted partition selection algorithms: MAD2R \cite{chen2025scalable} and PolicyGaussian \cite{gopi20a}. Adopting our subroutine improves performance of both of these approaches, yielding state-of-the-art approaches for partition selection in both the parallel and sequential regimes.

\paragraph{Optimal partition selection with additive noise (\Cref{sec:additive-partition-selection}).} Motivated by the fact that \cite{desfontaines-partition} nearly matches utility with the truncated discrete Laplace mechanism, we present a convex program to find the optimal partition selection mechanism when a) users only submit a single partition, and b) when the algorithm must be the post-processing of an additive noise mechanism. We find a numerical separation in privacy between additive and non-additive approaches to partition selection. Since one benefit of using additive noise is that we privatize the count as well as the partition itself simultaneously, this result can be seen as the ``cost of releasing the count" associated with any partition. When the count is not needed, our result demonstrates that additive-noise based techniques are (unfortunately) sub-optimal.
\section{Preliminaries} \label{sec:prelim}

For every $K \in \N$, we use $[K]$ to denote the set $\{1, \dots, K\}$. For a given distribution $P$, we will write $P_+$ to denote the shift of $P$ by one. That is, $P(x) = P_+(x-1)$.

Let $\Ber(p)$ denote the Bernoulli distribution. For $X \sim \Ber(p)$, $\bbP(X = 1) = p = 1 - \bbP(X = 0)$. 



\subsection{Privacy tools}

We will begin by defining the relevant privacy notions used in this work. Below, we define differential privacy (DP) under a generic neighboring notion. Specific neighboring notions relevant to our work will be introduced when we formalize the problem in \Cref{subsec:prelim-partition-selection}.

\begin{definition}[Differential privacy \cite{dwork-calibrating}] \label{def:dp}
A randomized algorithm $M: \mathcal{X} \to \mathcal{Y}$ satisfies \emph{$(\eps, \delta)$-differential privacy ($(\eps, \delta)$-DP)} if, for all neighboring inputs $X, X' \in \mathcal{X}$
, and for all $S \subseteq \mathcal{Y}$, $
\mathbb{P}(M(X) \in S) \le e^\eps \cdot \mathbb{P}(M(X') \in S) + \delta
$.
\end{definition}

As mentioned earlier, we are interested in variants of DP that uses (approximate) R{\'e}nyi divergence to measure the divergence between $M(X)$ and $M(X')$. We start by recalling R{\'e}nyi divergence and its approximate counterpart below.

\begin{definition}[R{\'e}nyi divergence \cite{renyi61}]
Let $P$ and $Q$ be probability distributions on $\Omega$. Then the R{\'e}nyi divergence between $P$ and $Q$ at order $\alpha$, denoted $\dr{P}{Q}$, is 

$$
\dr{P}{Q} = \frac{1}{\alpha-1} \log\left(\int_\Omega P(x)^\alpha Q(x)^{1-\alpha} \mathrm{d}x\right),
$$

where $P(\cdot), Q(\cdot)$ denote the probability mass/density functions of $P$ and $Q$, respectively.\footnote{If $P$ is not absolutely continuous with respect to $Q$, we define $\dr{P}{Q} = \infty$ for all $\alpha \ge 1$.} The R{\'e}nyi divergence at $\alpha = 1$ is defined as 

$$
\dra{P}{Q}{1} = \lim_{\alpha \to 1} \dr{P}{Q} = D_{KL}(P || Q) = \int_{\Omega} P(x) \log \left(\frac{P(x)}{Q(x)} \right)\mathrm{d}x.
$$
\end{definition}

\begin{definition}[Approximate R{\'e}nyi divergence \cite{papernot2022hyperparameter}]\label{def:approx-renyi-divergence}
Let $P$ and $Q$ be probability distributions over $\Omega$ and $\delta \in [0, 1]$, then
\begin{align*}
\drp{P}{Q} = \inf\left\{\dr{P'}{Q'} : P = (1-\delta)P' + \delta P'', Q=(1-\delta)Q' + \delta Q''\right\}
\end{align*}

i.e. $P$ and $Q$ can both be expressed as a convex combination of two distributions with weights $1-\delta$ and $\delta$, respectively.

\end{definition}

For notational convenience, we will sometimes write random variables in place of their distributions in $\dr{\cdot}{\cdot}$ when there is no ambiguity.

The privacy notion of our interest can now be defined as follows:

\begin{definition}[Approximate R{\'e}nyi DP \cite{papernot2022hyperparameter}] \label{def:approx-rdp}
A randomized algorithm $M: \mathcal{X} \to \mathcal{Y}$ satisfies \emph{$\delta$-approximate $(\alpha, \hat{\eps})$-R{\'e}nyi differential privacy} (denoted $(\delta, \alpha, \eps)$-RDP for brevity) if, for all neighboring inputs $X, X' \in \mathcal{X}$
, $\drp{M(x)}{M(x')} \le \hat{\eps}$.
\end{definition}

Standard R{\'e}nyi differential privacy \cite{Mironov17} is a special case of approximate RDP with $\delta = 0$, which it shares the following straightforward composition properties with.

\begin{proposition}[Approximate RDP composition] \label{prop:apx-rdp-composition}
Let $M_1$ and $M_2$ satisfy $(\delta_1, \alpha, \eps_1)$ and $(\delta_2, \alpha, \eps_2)$ RDP, respectively. Then $M_1 \circ M_2$ satisfies $(\delta, \alpha, \eps_1 + \eps_2)$ RDP, where $\delta = \delta_1 + \delta_2 - \delta_1 \cdot \delta_2 \le \delta_1 + \delta_2.$

In particular, if we have distributions $P_1, P_2, Q_1, Q_2$ such that $\drpd{\delta_1}{P_1}{Q_1} \leq \eps_1$ and $\drpd{\delta_2}{P_2}{Q_2} \leq \eps_2$, then $\drpd{\delta_1 + \delta_2}{P_1 \otimes P_2}{Q_1 \otimes Q_2} \leq \eps_1 + \eps_2$.
\end{proposition}

Converting from approximate RDP to approximate DP is a simple extension of \cite{canonne2020discrete} to add the extra $\delta$ term.

\begin{proposition}\label{prop:approx-rdp-to-dp}
Let $M$ satisfy $(\delta, \alpha, \eps)$-RDP. Then $M$ satisfies $(\hat{\delta}, \hat{\eps})$-DP for 
$$
\hat{\delta} = \delta + \frac{\exp((\alpha - 1)(\eps - \hat{\eps}))}{\alpha}\cdot \left(1 - \frac{1}{\alpha}\right)^{\alpha-1}.
$$
\end{proposition}

\subsection{Partition selection}
\label{subsec:prelim-partition-selection}

In this section we will formalize the problem of partition selection.
Let $U$ be a universe of elements. A dataset $X$ is represented as a vector in $\Z_{\geq 0}^U$, where $X_u$ is the number of occurrences of item $u$. 

\begin{definition}[Private partition selection]
A partition selection mechanism $M: \Z_{\geq 0}^U \to 2^U$ takes as input a database and returns as many partitions as possible in the database (under the privacy constraint). In other words, the output $M(X)$ must be a subset of $\{u \in U \mid X_u > 0 \}$.
\end{definition}

We say that two datasets $X, X'$ are neighbors under $L^r$ norm bound $\Delta$ if $\|X - X'\|_r \leq \Delta$. 

We focus on differentially private algorithms based on applying a \emph{partition selection primitive} independently on each element, as formalized below.

\begin{definition}[Differentially private partition selection primitive]\label{def:partition-selection-prim}
A partition selection primitive is a function $\pi : \mathbb{Z}_{\ge 0} \to [0, 1]$ where $\pi(0)= 0$. Its corresponding partition selection mechanism $M_{\pi}$ works as follows: For every $u \in U$, include $u$ independently in the output with probability $\pi(X_u)$.

We say that $\pi$ is $(\delta, \alpha, \eps)$-RDP for $\Delta_r = \Delta$ if $M_{\pi}$ is $\delta$-approximate $(\alpha, \eps)$-RDP under the neighboring notion for $L^r$ norm bound $\Delta$.
\end{definition}

We similarly define \emph{weighted partition selection} and \emph{weighted partition selection primitive} (denoted by a release probability function $\phi: \R_{\geq 0} \to [0, 1]$) identically as above, except that each dataset is now represented as a vector of non-negative \emph{real numbers}, i.e. $X \in \R^U_{\geq 0}$. Here $X_u$ denote the total weight of item $u$. 
Moreover, we also sometimes allow multiple norm bounds.


\begin{remark}
For notational convenience, our ``dataset'' above already consists of total counts or weights of elements. However, this can also be a result of a function (e.g. policy from \cite{gopi20a}) of an underlying database. As long as this function results has bounded $L_r$-sensitivity, we can apply our mechanisms.
\end{remark}

Finally, we outline what we mean by \emph{optimality} of partition selection. 

\begin{definition}\label{def:prim-opt}
A partition selection primitive $\piopt$ is optimal for $(\delta, \alpha, \eps)$-RDP with $\Delta_r = \Delta$ if it satisfies $(\delta, \alpha, \eps)$-RDP for $\Delta_r = \Delta$ and, additionally, for any other partition selection primitive $\pi$ satisfying the same privacy constraint and all $n \in \mathbb{N}$, $\piopt(n) \ge \pi(n)$.
\end{definition}

\begin{definition}
    A partition selection mechanism $M$ is optimal for $(\delta, \alpha, \eps)$-RDP with $\Delta_r = \Delta$ if it satisfies $(\delta, \alpha, \eps)$-RDP for $\Delta_r = \Delta$ and additionally, for any other partition selection mechanism $M'$ satisfying the privacy constraint: 
    $E[|M(X)|] \ge E[|M'(X)|]$ for all datasets $X$.
\end{definition}

\section{Optimal partition selection}\label{sec:main-result}

Like \cite{desfontaines-partition}, we will start by focusing on the unweighted case where each user only contributes a single element to a single partition, i.e.\footnote{Since this is the unweighted case, we can also set $\Delta_r = 1$ for any $1 < r < \infty$ (or $\Delta_0 = \Delta_\infty = 1$) instead.} $\Delta_1 = 1$. 
In this case, we will describe a simple formula for the optimal partition selection primitive $\piopt$ for $(\delta, \alpha,\eps)$-RDP.

Before do so, it will be useful to describe how the approximate R{\'e}nyi divergence between two Bernoulli random variables behaves.

\begin{lemma}\label{bernoulli-monotonic}
    Let $P \sim \Ber(p), Q \sim \Ber(q)$ and $\alpha > 1$. Fixing $p$, $\dr{P}{Q}$ is decreasing in $q$ on the interval $[0, p]$ and increasing in $q$ on the interval $[p, 1]$. Similarly, for a fixed $q$, $\dr{P}{Q}$ is decreasing in $p$ in $[0, q]$ and increasing in $p$ on $[q, 1]$.
\end{lemma}
\begin{proof}
Clearly, $\dr{P}{Q}$ is minimized at 0 when $p = q$. The result follows from the joint quasi-convexity of R{\'e}nyi divergence for $\alpha \ge 0$ (see \cite{van2014renyi}).
\end{proof}

\begin{lemma}\label{bernoulli-divergence}
Let $P \sim \Ber(p), Q \sim \Ber(q)$, and $\delta \in [0, 1)$.
Then 
$$
\drp{P}{Q} =
\begin{cases}
\dr{\Ber\left(\frac{p}{1-\delta}\right)}{\Ber\left(\frac{q - \delta}{1-\delta}\right)} & p < q - \delta\\
\dr{\Ber\left(\frac{p - \delta}{1- \delta}\right)}{\Ber\left(\frac{q}{1-\delta}\right)} & p > q + \delta\\
0 & |p - q| \le \delta.
\end{cases}
$$

Furthermore, fixing $p$, $\drp{P}{Q}$ is decreasing in $q$ on the interval $[0, p]$ and increasing in $q$ on the interval $[p, 1]$. Similarly for a fixed $q$, $\drp{P}{Q}$ is decreasing in $p$ on $[0, q]$ and increasing in $p$ on $[q,1]$.
\end{lemma}
\begin{proof}
From \Cref{def:approx-renyi-divergence}, we want to find distributions $P' \sim \Ber(p'), Q' \sim \Ber(q')$ that minimize $\dr{P'}{Q'}$ such that $P = (1-\delta)P' + \delta P''$ and $Q = (1-\delta)Q' + \delta Q''$. From \Cref{bernoulli-monotonic}, it suffices to move $p'$ and $q'$ closer to equality. The maximum we can increase $p'$ to is clearly $p / (1-\delta)$ by setting $P'' = 0$, and the minimum is $(p - \delta )/(1 - \delta)$ by setting $P'' = 1$. The same is true for $Q$ and the first result naturally follows by observing in cases where $p'$ and $q'$ would otherwise cross, they can be made equal instead.

The second observation follows directly from applying \Cref{bernoulli-monotonic} to each term.
\end{proof}

\begin{theorem}\label{thm:opt-partition-prim}
    Let 
    $$
    L(q) = \max\left\{
        p \in [q, 1]: 
            \drp{\Ber(p)}{\Ber(q)} \le \eps \land
            \drp{\Ber(q)}{\Ber(p)} \le \eps
    \right\}
    $$
    and 
    $$
    \piopt(n) = \begin{cases}
        0 & n = 0\\
        L(\piopt(n-1)) & n > 0
    \end{cases}.
    $$
    Then $\piopt(n)$ is the optimal partition selection primitive for $(\delta, \alpha, \eps)$-RDP with $\Delta_1 = 1$.
\end{theorem}
\begin{proof}
\emph{(Privacy)} Consider any two neighbors $X, X'$ whose counts of a partition $u \in U$ are $n, n'$ respectively (and the counts of other items are the same) where $n = n' - 1$. 
In either case, by construction we have $$\drp{M_{\pi^*}(X)}{M_{\pi^*}(X')} = \drp{\Ber(\piopt(n))}{\Ber(\piopt(n - 1))} \le \eps,$$ 
and, $$\drp{M_{\pi^*}(X')}{M_{\pi^*}(X)} = \drp{\Ber(\piopt(n-1))}{\Ber(\piopt(n))} \le \eps.$$
Thus, $M_{\pi^*}$ is $\delta$-approximate $(\alpha, \eps)$-RDP as desired.

\emph{(Optimality)} First we will show that $L$ is increasing in $q$. 
Fix any $p > q$. By \Cref{bernoulli-divergence}, the divergence terms $\drp{\Ber(p)}{\Ber(q)}$ and $\drp{\Ber(q)}{\Ber(p)}$ are decreasing as $q$ increases towards $p$.
Therefore as $q$ increases, larger values of $p$ satisfy the divergence constraints.

Now consider any partition selection primitive $\pi$ where $\pi(n_0) > \piopt(n_0)$; let $n_0$ be the smallest such index, so that $\pi(n_0 - 1) \le \piopt(n_0 - 1)$. Since $\pi^*(n_0) = L(\pi^*(n_0 - 1))$, the monotonicity of $L$ implies that $\pi(n_0) > L(\pi(n_0 - 1))$. In other words, $\drp{\Ber(\pi(n_0))}{\Ber(\pi(n_0 - 1))} > \eps$ or $\drp{\Ber(\pi(n_0 - 1))}{\Ber(\pi(n_0))} > \eps$. This means that $M_{\pi}$ is not $\delta$-approximate $(\alpha, \eps)$-RDP.
\end{proof}

\Cref{alg:piopt} computes $\piopt$ iteratively. It reduces the problem to computing the value of $p$ which satisfies an exact (pure) RDP guarantee between $\Ber(p)$ and a fixed $\Ber(q)$. Both $p_1$ and $p_2$ can be solved and guarantee convergence with simple bisection techniques, as they reduce to solving bounded convex and quasi-convex minimization problems, respectively. 

\begin{algorithm}
\caption{Procedure to compute $\piopt(n)$}\label{alg:piopt}
\begin{algorithmic}
\Require $n \in \mathbb{Z}_{\ge 0}$, $\eps > 0$, $\alpha > 1$, $\delta \in (0, 1)$
\State $\pi \gets 0$
\Loop\ $n$ times
    \If{$\pi + \delta \ge 1$} \textbf{ return } 1
    \EndIf
    \State $q = \min\left\{1, \frac{\pi}{1 - \delta}\right\}$
    \State $p_1 \gets \left\{p \in [q, 1]: \dr{\Ber(p)}{\Ber(q)} = \eps\right\}$
    \State $p_2 \gets \left\{p \in [q, 1]: \dr{\Ber(q)}{\Ber(p)} = \eps\right\}$
    \State $p = \min(p_1 \cup p_2)$    
    \State $\pi \gets \min\left\{1, p + \delta - p \delta \right\}$
\EndLoop
\State \textbf{return } $\pi$
\end{algorithmic}
\end{algorithm}

\subsection{General optimality of $\piopt$}

Here we will extend our optimality result to show $\piopt$ dominates mechanisms that do not necessarily consider each partition independently. Like, \cite{desfontaines-partition}, we restrict to the $\Delta_1 = 1$ case.

\begin{theorem}\label{thm:general-opt}
$M_{\pi^*}$ is the optimal partition selection mechanism for $(\delta, \alpha, \eps)$-RDP with $\Delta_1 = 1$
\end{theorem}
\begin{proof}

Consider a single partition $u$, and denote $f(u) = \Pr(u \in M(X))$ the probability of releasing $u$ under an arbitrary mechanism $M$ under $(\delta, \alpha, \eps)$-RDP. Fixing all other partitions, we have $f(u) \le \piopt(X_u)$ by \Cref{thm:opt-partition-prim}. Thus,
$
\E[|M(X)|] = \sum_{u \in U} f(u) \le \sum_{u \in U} \piopt(X_u) = \E[|M_{\pi^*}(X)|].
$
\end{proof}


\subsection{Non-existence of optimal partition selection when $\Delta_1 \neq 1$} \label{subsec:non-opt}

Unlike the case $\Delta_1 = 1$, we show below that, when $\Delta_1 \ne 1$, the optimal selection mechanism does \emph{not} exist for certain regime of parameters.

\begin{theorem}\label{thm:no-optimal-multi}
For any $0 < \alpha, \eps, \delta$ such that\footnote{Here $\pi^*$ is as defined in \Cref{thm:opt-partition-prim} for $\Delta_1 = 1$.}\footnote{It is simple to see that the inequality $\pi^*(2) > 3 \cdot \pi^*(1)$ holds for all sufficiently large $\eps$ and sufficiently small $\delta$.} $\pi^*(2) > 3 \cdot \pi^*(1)$, there is no optimal selection mechanism even for $\Delta_1 = 2$. 
\end{theorem}

\begin{proof}
Suppose for the sake of contradiction that there exists an optimal selection mechanism $M^*$ for $(\delta, \alpha, \eps)$-RDP with $\Delta_1 = 2$ and universe $U = \{1, 2\}$. Consider neighboring datasets $X^1 = (1, 1)$ and $X^0 = (0, 0)$. Notice that, by comparing $\bbP(M^*(X^1) \ne \emptyset)$ and $\bbP(M^*(X^0) \ne \emptyset) = 0$, we have
\begin{align} \label{eq:tmp-emptyset-lb}
\bbP(M^*(X^1) \ne \emptyset) \leq \pi^*(1).
\end{align}
We consider two cases, based on $\bbP(M^*(X^1) = U)$.

\paragraph{Case I: $\bbP(M^*(X^1) = U) < \pi^*(1)$.} In this case, consider another mechanism $M'$ where $\bbP(M'(X^1) = U) = \pi^*(1), \bbP(M'(X^1) = \emptyset) = 1 - \pi^*(1)$, and $\bbP(M'(X) = \emptyset) = 1$ for all $X \ne X^1$. It is simple to check that $M$ is $\delta$-approximate $(\alpha, \eps)$-RDP and that $\E[|M^*(X^1)|] < \E[|M'(X^1)|]$, a contradiction.

\paragraph{Case II: $\bbP(M^*(X^1) = U) \ge \pi^*(1).$} Note that from \eqref{eq:tmp-emptyset-lb}, we must have $\bbP(M^*(X^1) = U) = \pi^*(1)$ and $\bbP(M^*(X^1) = \emptyset) = 1 - \pi^*(1)$. Consider another dataset $X^2 = (3, 1)$, which is a neighbor of $X^1$. By comparing $\bbP(M^*(X^2) \notin \{U, \emptyset\})$ with $\bbP(M^*(X^1) \notin \{U, \emptyset\}) = 0$, we have 
\begin{align*}
\bbP(M^*(X^2) \notin \{U, \emptyset\}) \leq \pi^*(1).
\end{align*}
Consider yet another dataset $X^3 = (3, 0)$, which is a neighbor of $X^2$. By the guarantee of partition selection, we must have $\bbP(M^*(X^3) = U) = 0$
Similarly, by comparing $\bbP(M^*(X^2) = U)$ with $\bbP(M^*(X^3) = U) = 0$, we have 
\begin{align*}
\bbP(M^*(X^2) = U) \leq \pi^*(1).
\end{align*}
Together, this implies that
\begin{align*}
\E[|M^*(X^2)|] &= 1 \cdot \bbP(M^*(X^2) = \{1\}) + 1 \cdot \bbP(M^*(X^2) = \{2\}) + 2 \cdot \bbP(M^*(X^2) = U) \\
&= \bbP(M^*(X^2) \notin \{U, \emptyset\}) + 2 \cdot \bbP(M^*(X^2) = U) \\
&\leq 3 \pi^*(1) \\
&< \pi^*(2).
\end{align*}
Meanwhile, consider a mechanism $M''$ that ignores all elements of $U$ except for $1$ and include $1$ in the output with probability $\pi^*(\lceil i/2 \rceil)$. This mechanism satisfies $\delta$-approximate $(\alpha, \eps)$-RDP and $\E[|M''(X^2)|] = \pi^*(2)$, which contradicts the optimality of $M^*$.
\end{proof}

It remains an interesting open question to extend the above non-optimality result to a larger regime of parameters $\eps, \delta$.
\section{Weighted partition selection and the SNAPS mechanism}\label{sec:weighted-selection}

In this section we consider the more general weighted partition selection problem, where users hold a weight associated with each partition. We introduce a new mechanism: SNAPS (Smooth Norm-Aware Partition Selection). We remark that, due to the non-optimality result from \Cref{subsec:non-opt}, there is no optimal mechanism in this setting. Therefore, we aim to design an algorithm that provides a good utility and is practical.

Our idea is to derive a weighted partition selection primitive which affords users a ``smooth'' privacy loss depending on how much weight they hold. This allows us to derive a composition-based mechanism where users can hold multiple partitions with a total bound on the $L^r$ weight norm.

To define the primitive, we let
$$
    L(q, \eps, \delta) = \max\left\{
        p \in [q, 1]: 
            \drp{\Ber(p)}{\Ber(q)} \le \eps \land
            \drp{\Ber(q)}{\Ber(p)} \le \eps
    \right\}.
$$

\begin{definition}[Weighted partition selection primitive]\label{def:weighted-selection-prim}
Given parameters $\eps_0, \delta_0, r, \eps_1, \delta_1,\disc, \Delta > 0$, first we define the \emph{discretized weighted partition selection primitive} $\discprim_r$ over non-negative integers as follows: 
\begin{itemize}
\item For $n = 0$, let $\discprim_r(0) = 0$
\item For $n > 0$, let $\ndisc = \left\lceil \frac{\Delta}{\disc} \right\rceil$ and
$$\discprim_r(n) = \min_{i \in [\min\{n, \ndisc\}]} L\Paren{\discprim_{r}(n - i), \eps_0 + \eps_1 \cdot \Paren{\disc(i-1)}^r, \delta_0 + \delta_1 \cdot \Paren{\disc(i-1)}^r}$$
\end{itemize}
We define the \emph{weighted partition selection primitive} $\phi_r$ on non-negative real numbers as follows:
\begin{align*}
\phi_r(y) = \discprim_r\Paren{\left\lfloor \frac{y}{\disc} \right\rfloor}.
\end{align*}
\end{definition}

We next point out that the definition of weighted partition selection primitive almost immediately gives the privacy guarantee for $\Delta_0 = 1$ setting, as shown below.\footnote{We note that if we impose an additional constraint on the \emph{minimum} weight value that a user can submit, we can avoid the need for $\eps_0, \delta_0$ params with a fine enough discretization. Furthermore, this additional restriction does indeed improve utility relative to \Cref{def:weighted-selection-prim}. We avoid this for sake of presenting an algorithm which can always replace the Gaussian mechanism for partition selection.}

\begin{lemma} \label{lem:weighted-one-item}
For $\Delta_0 = 1$ and $\Delta_\infty \leq \Delta$, the weighted partition selection primitive $\phi_r$ satisfies $(\delta_0 + \delta_1 \cdot \Delta_{\infty}^r, \alpha, \eps_0 + \eps_1 \cdot \Delta_{\infty}^r)$-RDP.
\end{lemma}

\begin{proof}
Consider any two neighboring datasets $X, X'$ which differs on a single item $u$. 
Assume w.l.o.g. that $X_u \leq X_u'$. By our assumption, we have $X_u' - X_u \leq \Delta_{\infty}$. Let $z = \left\lfloor\frac{X_u}{\disc}\right\rfloor$ and $z' = \left\lfloor\frac{X_u'}{\disc}\right\rfloor$. We also have $i := z' - z \leq \lceil \Delta_{\infty} / \disc \rceil \leq \ndisc$.

Let $\delta = \delta_0 + \delta_1 \cdot \Delta_\infty^r$ and $\eps = \eps_0 + \eps_1 \cdot \Delta_\infty^r$. We have
\begin{align*}
&\max\{\drp{M_{\phi_r}(X')}{M_{\phi_r}(X)}, \drp{M_{\phi_r}(X)}{M_{\phi_r}(X')}\} \\
&= \max\{\drp{\phi_r(X_u')}{\phi_r(X_u)}, \drp{\phi_r(X_u)}{\phi_r(X_u')}\} \\ &= \max\{\drp{\Ber{\Paren{\discprim_r(z)}}}{\Ber{\Paren{\discprim_r(z')}}}, \drp{\Ber{\Paren{\discprim_r(z')}}}{\Ber{\Paren{\discprim_r(z)}}}\}. 
\end{align*}
Finally, by definition of $\discprim_r(z') = \discprim_r(z+i)$, it is at most $L(\discprim_{r}(z), \eps_0 + \eps_1 \cdot \Paren{\disc(i-1)}^r, \delta_0 + \delta_1 \cdot \Paren{\disc(i-1)}^r) \leq L(\discprim_{r}(z), \eps, \delta)$, where the inequality follows from $i - 1 \leq \Delta_\infty / \disc$. Thus, by definition of $L$, the above term is bounded above by $\eps$ as desired.
\end{proof}

The above result can be easily extended to the case $\Delta_0 > 1$ with $L^r$ norm bound $\Delta_r$, with the cost being a $\Delta_0$ multiplicative factor in front of $\eps_0, \delta_0$, as formalized below.

\begin{theorem}[SNAPS Mechanism] \label{cor:snap}
For any $\Delta_0 \in \N$ and $\Delta_r \leq \Delta$, the weighted partition selection primitive $\phi_r$ satisfies $(\delta_0 \cdot \Delta_0 + \delta_1 \cdot \Delta_r^r, \alpha, \eps_0 \cdot \Delta_0 + \eps_1 \cdot \Delta_r^r)$-RDP. We refer to the associated weighted partition selection mechanism $M_{\phi_r}$ as the SNAPS mechanism.
\end{theorem}

\begin{proof}
Consider any pair of neighboring input datasets $X, X'$ with $\|X - X'\|_r \leq \Delta_r$ and $\|X - X'\|_0 \leq \Delta_0$. In particular, let $S \subseteq U$ denote the set of all $u \in U$ such that $X_u \ne X'_u$. Let $\eps = \eps_0 \cdot \Delta_0 + \eps_1 \cdot \Delta_r^r$ and $\delta = \delta_0 \cdot \Delta_0 + \delta_1 \cdot \Delta_r^r$. We have
\begin{align*}
\drp{M_{\phi_r}(X)}{M_{\phi_r}(X')} &= \drp{\bigotimes_{u \in S} \phi_r(X_u)}{\bigotimes_{u \in S} \phi_r(X'_u)}.
\end{align*}
Let $\delta_u = \delta_0 + \delta_1 \cdot |X_u - X'_u|^r$ and $\eps_u = \eps_0 + \eps_1 \cdot |X_u - X'_u|^r$.
From \Cref{lem:weighted-one-item}, we have $$\drpd{\delta_u}{\phi_r(X_u)}{\phi_r(X'_u)} \leq \eps_u.$$
Notice that $\sum_{u \in S} \delta_u \leq \delta$ and $\sum_{u \in S} \eps_u \leq \eps$. Thus, applying the composition theorem (\Cref{prop:apx-rdp-composition}), we can conclude that $\drp{M_{\phi_r}(X)}{M_{\phi_r}(X')} \leq \eps$. As a result, the mechanism is $\delta$-approximate $(\alpha, \eps)$-RDP.
\end{proof}

We find that the SNAPS mechanism can often be parameterized to outperform the Gaussian mechanism for weighted partition selection for $L^2$ bounded input, and it can replace the Gaussian mechanism in all cases when the noisy weight vector output is not needed.

\newpage
\section{Experiments}
\label{sec:exp}

\begingroup
\setlength{\columnsep}{ .5cm} 
\begin{wrapfigure}{r}{0.5\textwidth} 
\centering
\centering
\includegraphics[width=1 \linewidth]{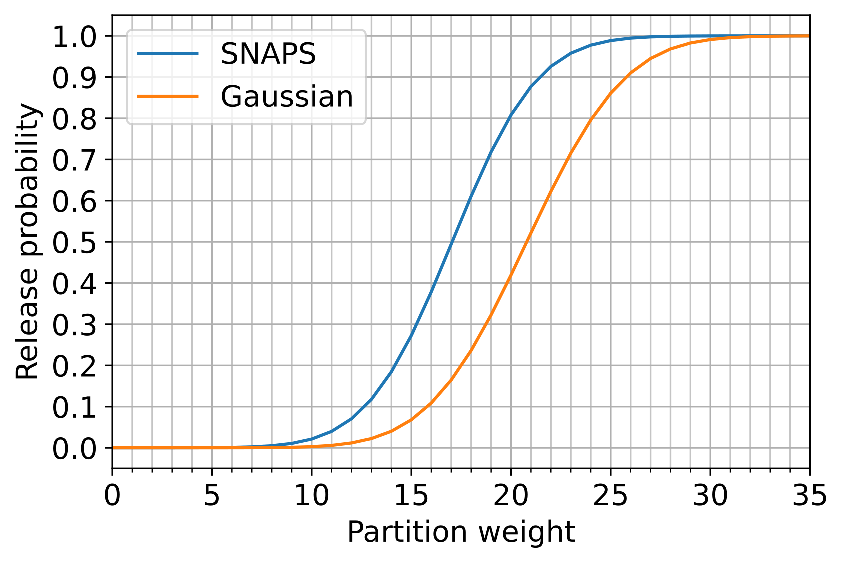} 
\captionsetup{font=footnotesize}
\caption{
Gaussian and SNAPS mechanism release probabilities for $(1, 10^{-5})$-DP. SNAPS is parameterized as in \Cref{sec:exp}. The Gaussian probabilities are determined following \cite{gopi20a,chen2025scalable}, i.e. by splitting the $\delta$ budget evenly in two, one for computing the analytical Gaussian $\sigma$, and one for determining the threshold $\tau = \max_{k \in [\Delta_0]} \{1/\sqrt{k} + \sigma \cdot \Phi^{-1}\left((1-\delta/2)^{1/k}\right)\}$.}
\label{fig:guass-vs-snaps}
\end{wrapfigure}

We showcase the versatility of the SNAPS mechanism by plugging it into two adaptive partition selection mechanisms, MAD2R \cite{chen2025scalable} and PolicyGaussian \cite{gopi20a}.\footnote{
We use \url{https://github.com/heyyjudes/differentially-private-set-union} and \url{https://github.com/jusyc/dp_partition_selection} for PolicyGaussian and MAD2R implementations, respectively.} Both algorithms use an adaptive process to optimize the weight vector associated with each user while maintaining an $L^2$ norm bound. We replace the Gaussian mechanism with the SNAPS mechanism into the very last step in both cases.

\paragraph{Datasets.} We consider a number of datasets commonly used in benchmarking differentially private partition selection. Reddit \cite{gopi20a} is a dataset of posts to \verb*|r/askreddit|, Wiki \cite{wiki} is a dataset of Wikipedia abstracts. Twitter \cite{twitter} is a dataset of customer support tweets, Finance \cite{finance} is a dataset of stock market news data, Amazon \cite{amazon1,amazon2} is a dataset of Amazon reviews, and IMDb \cite{imdb} is a dataset of movie reviews.

\endgroup

For each datasets, we replicate methodology in \cite{gopi20a, chen2025scalable}: each token in a document is a partition, and each document corresponds to a user. In datasets where actual users are tracked, we combine users' documents.

\paragraph{Algorithms and parameters.}
We match the benchmarking conditions of \cite{chen2025scalable} by setting $\Delta_0 = 100$, and ensuring all algorithms satisfy $(1, 10^{-5})$-DP.

\begin{itemize}
    \item \textbf{MAD2R \cite{chen2025scalable}.} We use a basic composition privacy split of $[.1, .9]$ across two rounds. We set $d_{max}=50,b_{min}=0.5, b_{max}=2, C_{lb}=1$, and  $C_{ub}=3$. The adaptive threshold is set to $\beta = 2$ standard deviations above the baseline threshold.
    \item \textbf{Policy \cite{gopi20a}.} We use only a single pass, and set the adaptive threshold to $\beta = 4$ standard deviations above the baseline threshold.
    \item \textbf{SNAPS.} When integrating SNAPS into the above mechanisms, we match all of the above parameters. However, rather than setting the adaptive threshold as a function of the Gaussian noise (and $\beta$), we set it so that the \emph{release probability} under SNAPS is equalized with the Gaussian mechanism at that threshold (i.e. for $\beta = 2$ we choose a threshold that will release with probability $\sim 95\%$ under SNAPS). Furthermore in all cases we set\footnote{While these parameters were generally selected for good utility, we did not perform an exhaustive hyperparameter search. In particular, we use the exact same parameters for the $(.9, .9 \cdot 10^{-5})$-DP stage of MAD2R.} $\alpha = 18.5, \eps_0 = 10^{-5}, \delta_0 = 5 \cdot 10^{-5}, \Delta^{\text{disc}}= 5\cdot 10^{-4}$, though of course it yields different release probabilities than plotted in \Cref{fig:guass-vs-snaps}.
\end{itemize}

To translate approximate RDP guarantees from our algorithm to approximate DP, we leverage \Cref{prop:approx-rdp-to-dp} as implemented by \cite{DPLib}. In particular, to match a $(\hat{\eps}, \hat{\delta})$-DP guarantee, we ensure our SNAPS mechanism satisfies $(\hat{\delta}/2, \alpha, \eps)$-RDP such that the $(\alpha, \eps)$-RDP guarantee gives $(\hat{\eps}, \hat{\delta}/2)$-DP. In other words, we use half our $\delta$ budget in RDP $\to$ DP conversion. 

In \Cref{fig:guass-vs-snaps}, we plot the release probabilities between the Gaussian mechanism and SNAPS for the above parameterizations. As predicted by the plot, we should expect the SNAPS experiment variants to always outperform the Gaussian variants, which is validated by the results in  \Cref{tab:experiment_results}, showing that SNAPS improves utility (as measured by the output size) by 10-20\% in all cases.

\begin{table}[htbp]
    \centering
    \begin{tabular}{l c c c c}
        \toprule
        \textbf{Dataset} & \textbf{PolicyGaussian} & \textbf{PolicySNAPS} & \textbf{MAD2R} & \textbf{MAD2R-SNAPS} \\
        \midrule
Reddit & 7161 {\tiny ($\pm$ 10)} & \underline{8486} {\tiny ($\pm$ 22)} & 6187 {\tiny ($\pm$ 22)} & \textbf{7161} {\tiny ($\pm$ 20)}\\
Wiki & 11467 {\tiny ($\pm$ 22)} & \underline{13793} {\tiny ($\pm$ 47)} & 10431 {\tiny ($\pm$ 42)} & \textbf{12170} {\tiny ($\pm$ 39)}\\
Twitter & 15814 {\tiny ($\pm$ 45)} & \underline{18407} {\tiny ($\pm$ 43)} & 13847 {\tiny ($\pm$ 29)} & \textbf{15739} {\tiny ($\pm$ 53)}\\
Finance & 20129 {\tiny ($\pm$ 31)} & \underline{22952} {\tiny ($\pm$ 30)} & 17695 {\tiny ($\pm$ 48)} & \textbf{19969} {\tiny ($\pm$ 33)}\\
Amazon & 77840 {\tiny ($\pm$ 95)} & \underline{89416} {\tiny ($\pm$ 71)} & 65273 {\tiny ($\pm$ 91)} & \textbf{73703} {\tiny ($\pm$ 107)}\\
IMDB & 3582 {\tiny ($\pm$ 21)} & \underline{4447} {\tiny ($\pm$ 18)} & 3829 {\tiny ($\pm$ 23)} & \textbf{4559} {\tiny ($\pm$ 17)}\\
        \bottomrule
    \end{tabular}
    \caption{Comparison of output size of partition selection algorithms, with $\eps = 1, \delta = 10^{-5}$, and $\Delta_0 = 100$. Results for each algorithm are averaged over 5 trials, with one standard deviation reported. The best parallel result is bolded, and the best sequential result is underlined.}
    \label{tab:experiment_results}
\end{table}
\section{Additive noise for partition selection}\label{sec:additive-partition-selection}

After exploring weighted partition selection in \Cref{sec:weighted-selection}, we now pivot back to the notion of \emph{optimal} partition selection. A fundamental limitation of our optimal algorithm defined by $\pi^*$ (\Cref{thm:general-opt}) is that, unlike the Laplace or Gaussian mechanism, it does not allow releasing the \emph{total weight} of released partitions along with the set of released partitions. We observe that all mechanisms which add noise to the weight vector and then perform thresholding as a post-processing step allow releasing the weight as a byproduct, as long as the additive noise alone suffices to guarantee privacy. This is formalized below.

\begin{definition}
An additive-noise partition selection mechanism $M_{P, \tau}(n)$ is 
parameterized by a truncated probability distribution $P$ with support\footnote{
Defining additive-noise partition selection mechanisms this way precludes using the standard discrete Laplace / Gaussian mechanism (i.e. they need to be truncated). This is done to both simplify our subsequent optimization technique, and also to naturally define the class of mechanisms as pure post-processing of additive noise. This is fortunately without much loss of generality for our purposes (optimality for $\Delta_0 = 1$), since truncation on the right followed by normalization preserves privacy and decreases the release probability by at most $\delta$.
} on $Y \subseteq (-\infty, \tau] \subseteq \mathbb{Z}$. $M_{P,\tau}(n)$ samples $Z \sim P$ and releases $n + Z$, then keeps the partition if $n + Z > \tau$. The associated partition selection primitive is $\pi_{P, \tau}(n) = \bbP_{Z \sim P}(n + Z > \tau)$.
\end{definition}

Since $\pi_{P, \tau}(0) = \bbP_{Z \sim P}(Z > \tau) = 0$, $\pi_{P, \tau}$ is a valid partition selection primitive. Crucially, it is also a \emph{post-processing} of an additive noise mechanism. As such, for privacy it suffices to bound the loss associated with releasing $Z + q(D)$ where $q(D)$ is a sensitivity 1 query. This approach is precisely what allows us to release the noisy count.

The property of releasing the total partition weight ``for free" is a valuable aspect of these additive noise mechanisms. For instance, a private GROUP BY followed by a SUM in private SQL can be implemented with a single additive mechanism which simultaneously releases the partitions and their sums. As such, it would be very desirable if additive noise can come close to matching the privacy of $\pi^*$ for the same level of ``utility''.

In the $\alpha \to \infty$ regime explored in $\cite{desfontaines-partition}$, the truncated discrete Laplace mechanism was shown to be nearly optimal in most regimes, and exactly optimal for certain settings of $\eps$ and $\delta$. In this section, we show that there is an unfortunate inherent separation in privacy between additive mechanisms and the optimal mechanism in the $\alpha < \infty$ regime.
Throughout section, we follow the same restrictions as \Cref{sec:main-result} (and \cite{desfontaines-partition}) and only consider the $\Delta_1 = 1$ case.



\subsection{Approximate RDP of truncated additive noise mechanisms}\label{sec:additive-rdp}

Given a discrete distribution $P$ over the integers with finite support, our goal is to evaluate the approximate RDP guarantee of $P$ as an additive noise mechanism. Recall from \Cref{sec:prelim} that $P_+$ is the distribution of $P$ shifted by one.

It is simple to see--by following definitions--that the additive noise mechanism is $\delta$-approximate $(\alpha, \eps)$-RDP for $\eps = \max\{\drp{P_+}{P}, \drp{P}{P_+} \}$, and this is tight.

Before we proceed to any specific optimization problems related to additive noise mechanisms, let us consider the generic problem of computing $\drp{P}{Q}$ for arbitrary distributions $P, Q$.
For this problem, we first observe that the exponentiation R{\'e}nyi divergence and approximate R{\'e}nyi divergence is convex.

\begin{lemma}\label{lem:renyi-convex}
Let $P$ and $Q$ be distributions over some finite set $U$. Then for $\alpha > 1$, $e^{(\alpha - 1) \dr{P}{Q}} = E_{x \sim Q}\left[\left(\frac{P(x)}{Q(x)}\right)^\alpha\right]$ is jointly convex in $P$ and $Q$. 
\end{lemma}
\begin{proof}
Define $g(t) = t^\alpha$ which is convex for $\alpha > 1$ on $[0, \infty)$. The term inside the expectation is $f(u, v) = u^\alpha v^{1-\alpha}$. For $v > 0$ this is equal to $u \cdot g(u / v)$ which is the \emph{perspective function} of $g$ which is jointly convex in $v$ and $u$. Furthermore, this convexity is preserved for the edge case when $v = 0$ where $f(0, 0) = 0$ and $f(u, 0) = \infty$ for $u > 0$. The lemma statement follows due to the fact that expectations of jointly convex functions are jointly convex.
\end{proof}

\begin{proposition}\label{thm:approx-renyi-convex}
Let $P$ and $Q$ be distributions over some finite set $U$. Then for $\alpha > 1$, $e^{(\alpha -1 )\drp{P}{Q}}$ is jointly convex in $(P, Q, P', Q')$, where $P', Q'$ are the minimizing distributions in \Cref{def:approx-renyi-divergence}: $\drp{P}{Q} = \inf\left\{\dr{P'}{Q'} : P = (1-\delta)P' + \delta P'', Q=(1-\delta)Q' + \delta Q''\right\}
$.
\end{proposition}
\begin{proof}
Observe that by \Cref{def:approx-renyi-divergence}, $e^{(\alpha -1)\drp{P}{Q}}$ is the value of the following optimization problem:
\begin{align*}
    \text{minimize } & \sum_{x} P'(x)^\alpha Q'(x)^{1-\alpha} \\
    \text{subject to } & \sum P'(x) = \sum Q'(x) = 1, \\
                       & 0 \le P'(x) \le \frac{P(x)}{1-\delta}, \\
                       & 0 \le Q'(x) \le \frac{Q(x)}{1-\delta}.
\end{align*}
Let $f(P', Q') = \sum_{x} P'(x)^\alpha Q'(x)^{1-\alpha}$. By \Cref{lem:renyi-convex}, each term in the sum is jointly convex in $(P'(x), Q'(x))$, so $f$ is jointly convex in $(P', Q')$.
The constraints define a convex set in the joint space of $(P, Q, P', Q')$. Specifically, the inequality constraints are linear.
The function of interest is the partial minimization of a jointly convex function over a convex set, which preserves convexity.
\end{proof}

\paragraph{Fast Algorithm for Computing Approximate R{\'e}nyi Divergence.}
From \Cref{thm:approx-renyi-convex}, it is feasible to numerically solve for the approximate R{\'e}nyi divergence of two discrete distributions supported on a finite domain via standard convex optimization. However, there exists a greedy algorithm to compute this in $O(n \log n)$ time for any two discrete distributions $P$ and $Q$ with total support size $n$ using a ``water-filling" approach, as described below. 

Instead of $P', Q'$, it will be easier to describe an algorithm for finding $\tP = (1 - \delta)P'$ and $\tQ = (1 - \delta)Q'$. The algorithm works as follows:
\begin{enumerate}
    \item Compute $\tP$ by finding a cutoff $\lambda_P$ such that if we clip all likelihood ratios $\tP(x)/Q(x)$ to be at most $\lambda_P$, the total mass removed from $P$ is exactly $\delta$. More formally, find $\lambda_P$ such that
    $$
    \sum_x \tP(x) = 1-\delta \text{ for } \tP(x) = \min\{P(x), \lambda_P Q(x)\}
    $$
    \item Compute $\tQ$ by finding a cutoff $\lambda_Q$ such that if we clip all likelihood ratios $\tP(x)/\tQ(x)$ to be at least $\lambda_Q$, the total mass removed from $Q$ is exactly $\delta$. That is, find $\lambda_Q$ such that
    $$
    \sum_x \tQ(x) = 1-\delta \text{ for } \tQ(x) = \min\{Q(x), \tP(x)/\lambda_Q\}
    $$
\end{enumerate}

The approximate R{\'e}nyi divergence between $P$ and $Q$ is just the R{\'e}nyi divergence between the normalized versions of $P'$ and $Q'$. Note that $\lambda_P$ (resp. $\lambda_Q$) can be computed simply in $O(n \log n)$ time by iterating over the ratios $P(x) / Q(x)$ (resp. $\tP(x) / Q(x)$) in a sorted order, and manipulating the cumulative masses. 
More details and the proof of correctness are given in \Cref{app:waterfilling}.

It should be note that, remarkably, the distributions $P'$ and $Q'$ do \emph{not} depend on $\alpha$ at all. 

\subsection{Characterizing $\piopt$ as an additive noise mechanism.}

Recall that we have computed the optimal partition selection primitive $\pi^*$ in \Cref{sec:main-result}. In this section, we note that this mechanism can be viewed as an additive-noise partition selection strategy, and compute its privacy guarantees using the approach in the previous subsection.

\begin{definition}\label{def:bigpi}
For a given optimal partition selection primitive $\pi^*$, denote by $\Pi$ the discrete distribution with support on $\{0, \dots, n_d - 1\}$
where $n_d = \min\{n\in\mathbb{N} : \pi^*(n) = 1\}$ and
\begin{align*}
\Pi(x) = \pi^*(n_d - x) - \pi^*(n_d - 1 -x) & &\forall x \in \{0, \dots, n_d - 1\}
\end{align*}
\end{definition}
It is easy to see that $\pi^*$ is increasing and, for $\delta > 0$, $n_d$ is always finite (i.e. $n_d \leq \lceil 1/\delta \rceil$). Thus, $\Pi$ is a valid distribution. Furthermore, the additive-noise partition selection primitive from $\Pi$ with threshold $\tau = n_d - 1$ satisfies
\begin{align*}
\pi_{\Pi, \tau}(n) = \bbP_{Z \sim \Pi}(Z > \tau - n) = \sum_{x = \tau - n + 1}^{n_d - 1} \Pi(x) = \pi^*(n).
\end{align*}
In other words, the additive noise mechanism characterized by $\Pi$ \emph{precisely} matches the $\piopt$ strategy. Since $\piopt$ exactly characterizes $\Pi$ up to translation, no other additive-noise distribution can both 1) exactly match the utility of $\piopt$ and 2) improve privacy over $\Pi$.

We can numerically compute the privacy guarantee of the additive noise mechanism for $\Pi$ using the approach in \Cref{sec:additive-rdp}. \Cref{fig:additive-comp} shows there is a clear separation in privacy between $\piopt$ and $\Pi$
, which show a separation between any additive-noise partition selection strategies which \emph{exactly} match the selection probabilities of $\piopt$ in the $\alpha < \infty$ regime.

\subsection{Optimal additive noise with a single point-wise utility}

\begin{figure}[h]
\centering
\includegraphics[width=0.7\textwidth]{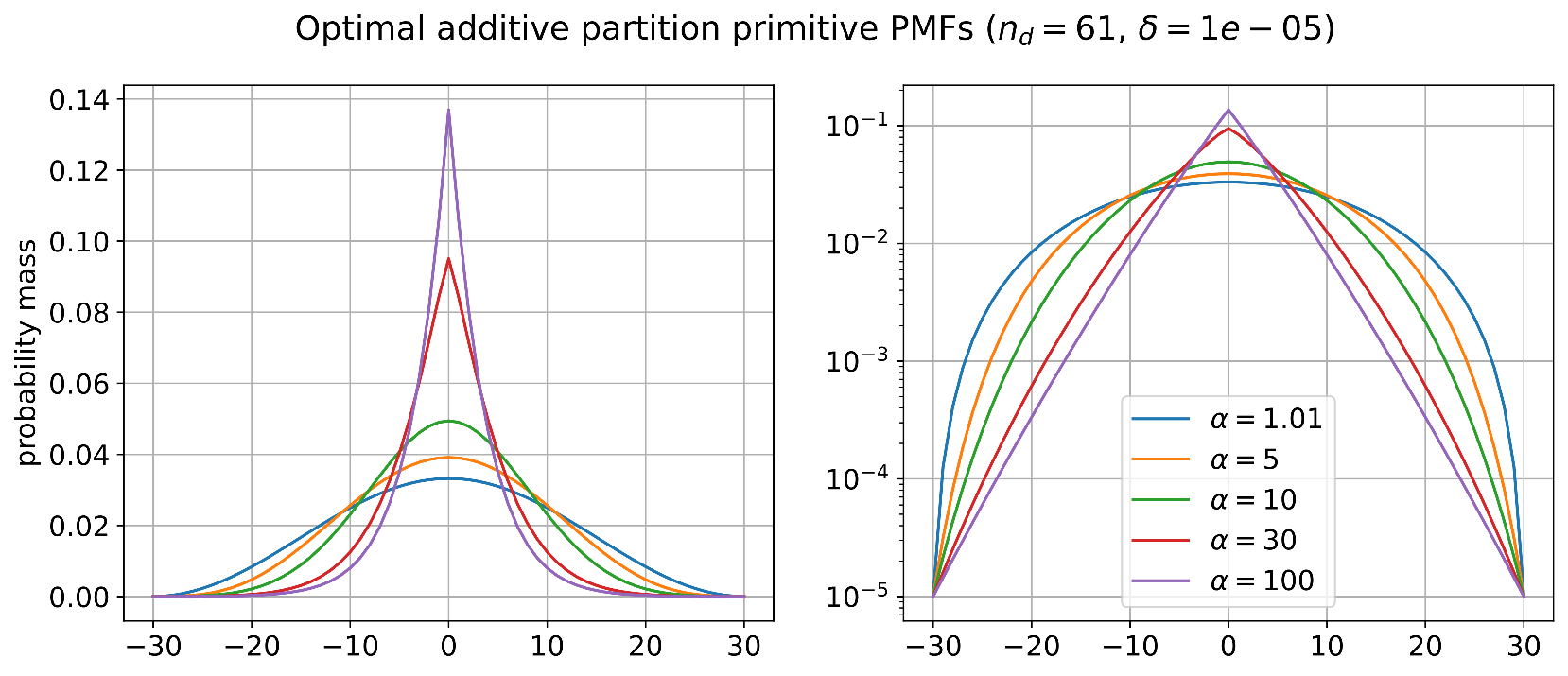}
\caption{Probability mass functions (centered at 0) for optimal additive noise distributions satisfying $\pi(61) = 1$ at various values of $\alpha$ as minimized by the convex program in \Cref{thm:opt-additive-simplified}. As $\alpha$ grows, the optimal distribution converges to a truncated discrete Laplace (\Cref{prop:converge-discrete-laplace}). At smaller $\alpha$ the optimal distribution becomes platykurtic, with a flatter peaks and thinner tails.}
\label{fig:additive-distributions}
\end{figure}

\begin{figure}[h]
\centering
\includegraphics[width=0.75\textwidth]{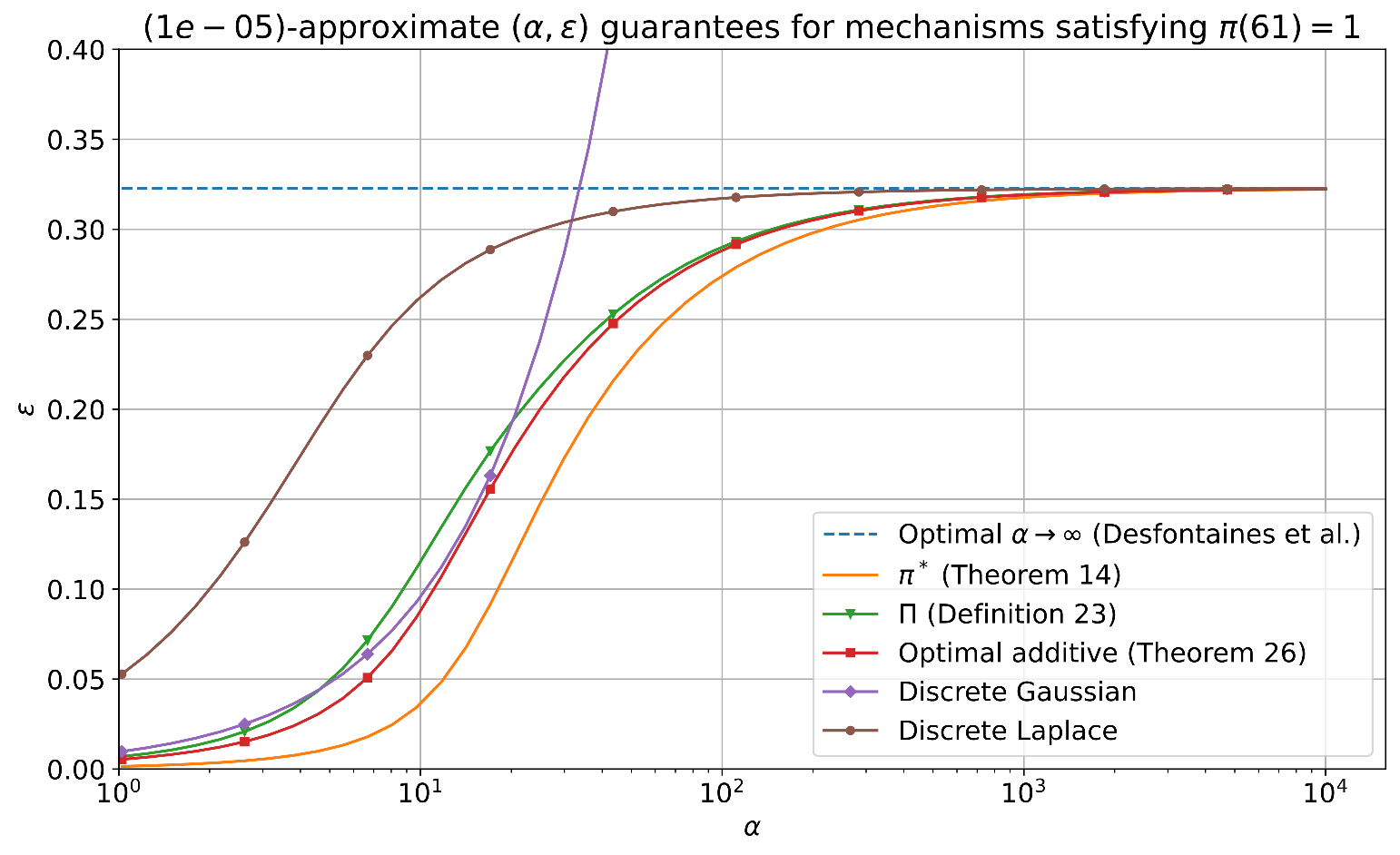}
\caption{The privacy of various mechanisms under the constraint that $\pi(61) = 1$ i.e. additive noise must be bounded in $[0, 60]$ (or equivalently $[-30, 30]$ due to translation invariance). $\pi^*$ clearly dominates all additive mechanisms for small and moderate $\alpha$. The (truncated) Gaussian and Laplace plots were computed by numerically solving for the scale parameters that ensure $f(-30)=f(30)=\delta$ in their respective PMFs after truncation and normalization.
}
\label{fig:additive-comp}
\end{figure}

In the previous section, we showed that there is a privacy gap if we \emph{exactly match} the selection probabilities of $\piopt$ with additive noise. In this section we will relax this to a specific point-wise guarantee: for a fixed $n_d \in \mathbb{N}$, we will find additive noise mechanisms which satisfy $\pi_{P,r}(n_d) = 1$ (i.e. those that release the partition with count at least $n_d$ deterministically\footnote{By restricting to $\pi(n_d) = 1$, we simplify the optimization problem substantially to a finite convex problem. In principle this can be extended to $\pi(n_d) = \gamma$ for any $\gamma > 0$, at some increase in complexity.}), while achieving the best possible privacy guarantee (minimizing $\eps$ for fixed $\delta, \alpha$).

Similar to \Cref{thm:approx-renyi-convex}, we can in fact find the optimal additive noise by formulating this as a convex optimization, but when $P$ is \emph{not} fixed apriori. The point-wise utility bound imposes a single constraint on $P$: that it has support only on $\{0, \dots, n_d - 1\}$. Under that constraint, it suffices to find the $P$ that minimizes the privacy objective $\max\{\drp{P}{P_+}, \drp{P_+}{P}\}$. This can be stated equivalently as follows.



\begin{align}
    \text{minimize } & \max\left\{\sum_{x} P_1'(x)^\alpha Q_1'(x)^{1-\alpha}, \sum_{x} Q_2'(x)^\alpha P_2'(x)^{1-\alpha}\right\} \label{eq:opt-additive}\\
    \text{subject to } & \sum P(x) = \sum P_i'(x) = \sum Q_i'(x) = 1, \nonumber \\
                       & 0 \le P_1'(x) \le \frac{P(x)}{1-\delta},\nonumber  \quad
                        0 \le Q_1'(x) \le \frac{P(x-1)}{1-\delta}.\nonumber\\ 
                       & 0 \le P_2'(x) \le \frac{P(x)}{1-\delta},\nonumber  \quad
                        0 \le Q_2'(x) \le \frac{P(x-1)}{1-\delta}.\nonumber 
\end{align}

Immediately we can show that under this optimization problem, we must have $P_1'(0) = Q_1'(0) = P_2'(n_d) = Q_2'(n_d) = 0$, otherwise either we explicitly violate the constraints, or the objective function is infinite. We can show even more structure to this problem though.

\begin{proposition}\label{prop:opt-symmetry2}
    The optimization problem to \Cref{eq:opt-additive} is symmetric in $P$. Namely, if $P$ is an optimal solution, then $\bar{P}$ defined by $\bar{P}(x) = P(n_d - 1 - x)$ is also an optimal solution.
\end{proposition}
\begin{proof}
Let $(P, P_1', Q_1', P_2', Q_2')$ be an optimal solution. Define $\bar{P}(x) = P(n_d - 1 - x)$ and 
    \begin{align*}
    \bar{P_1}'(x) &= Q_2'(n_d - x) \qquad
    \bar{Q_1}'(x) = P_2'(n_d - x) \\
    \bar{P_2}'(x) &= Q_1(n_d - x) \qquad
    \bar{Q_2}'(x) = P_1'(n_d - x).
    \end{align*}

The new objective on our transformed solution becomes
\begin{align*}
    &\max\left\{\sum_{x} \bar{P_1}'(x)^\alpha \bar{Q_1}'(x)^{1-\alpha}, \sum_{x} \bar{Q_2}'(x)^\alpha \bar{P_2}'(x)^{1-\alpha}\right\} \\
    &= \max\left\{\sum_{x} Q_2'(n_d - x)^\alpha P_2'(n_d - x)^{1-\alpha}, \sum_{x} P_1'(n_d - x)^\alpha Q_1(n_d - x)^{1-\alpha}\right\}\\
    &= \max\left\{\sum_{y} Q_2'(y)^\alpha P_2'(y)^{1-\alpha}, \sum_{y} P_1'(y)^\alpha Q_1(y)^{1-\alpha}\right\},
\end{align*}
which is identical to the original objective. It suffices to verify that the constraints are satisfied.

\begin{align*}
\bar{P_1}'(x) = Q_2'(n_d - x) \le \frac{P(n_d - 1 - x)}{1 - \delta} = \frac{\bar{P}(x)}{1 - \delta}\\
\bar{Q_1}'(x) = P_2'(n_d - x) \le \frac{P(n_d - x)}{1 - \delta} = \frac{\bar{P}(x - 1)}{1-\delta},
\end{align*}
and similarly the constraints on $(\bar{P_2}', \bar{Q_2}')$ hold by symmetry.
\end{proof}

When our problem is symmetrical, it gives us even more structure to exploit.

\begin{lemma}\label{lem:ardp-symmetric}
Let $P$ be a symmetric distribution and $Q = P + 1$ its shift by one. Then
$
D_\alpha^\delta(P || Q) = D_\alpha^\delta(Q || P).
$
\end{lemma}
\begin{proof}
The proof follows from showing that $(P, Q)$ is identical to $(Q, P)$ under a change of variables, i.e. a re-labeling of outcomes. Assume without loss of generality (up to a shift in coordinates) that $P$ is symmetric around 0 and $P(x) = P(-x)$. Now consider the change of variable $y = 1 - x$. Then we can show a bijection
\begin{align*}
P(y) &= P(1-x) = P(x - 1) = Q(x)\\
Q(y) &= P(-x) = P(x) &\qedhere
\end{align*}
\end{proof}

The above results allow us to simplify \Cref{eq:opt-additive} substantially.

\begin{theorem}\label{thm:opt-additive-simplified}
    For a fixed $\alpha > 1, \delta \in (0, 1), n_d \in \mathbb{N}$, let $P, P'$ be distributions on $\{0, \dots, n_d - 1\}$ and $\{1, \dots, n_d - 1\}$ respectively which  minimize the following convex program:
\begin{align*}
    \text{minimize } & \sum_{x} P'(x)^\alpha P'(n_d-x)^{1-\alpha}\\
    \text{subject to } & P(x) \ge 0,\\
    & \sum_x P(x) = \sum_x P'(x) = 1,\\
    &0 \le P'(x) \le \frac{P(x)}{1-\delta}\\
    & P(x) = P(n_d - 1 - x)\\
\end{align*}
Then $P$ minimizes $\eps$ in the $(\delta, \alpha, \eps)$-RDP guarantee among all distributions $X$ which satisfy $\pi_X(n_d) = 1$.
\end{theorem}
\begin{proof}

By \Cref{prop:opt-symmetry2} and convexity, there must exist an optimal solution to \Cref{eq:opt-additive} that is symmetric. Furthermore, by \Cref{lem:ardp-symmetric} it suffices to consider only one of the terms in the objective function.

We will show that the remaining problem is invariant upon reflecting and shifting $P_1'$ and $Q_1'$ in \Cref{eq:opt-additive}. Consider the following transformations:
    \begin{align*}
    \bar{P}'(x) &= Q_1'(n_d - x)\\
    \bar{Q}'(x) &= P_1'(n_d - x)
    \end{align*}
These map to the original problem constraints since  $P'(0) = Q'(0) = P'(n_d) = Q'(n_d) = 0$ is implied for any feasible solution due to the symmetry of $P$. Substituting $y = n_d - x$ yields

\begin{align*}
0 \le \bar{P}'(x) \le \frac{P(x)}{1- \delta} &\iff 0 \le Q_1'(y)\le \frac{P(y-1)}{1 - \delta} \\
0 \le \bar{Q}'(x) \le \frac{P(x - 1)}{1- \delta} &\iff 0 \le P_1'(y)\le \frac{P(y)}{1 - \delta}.\\
\end{align*}

Thus, if $(P, P_1', Q_1')$ is feasible, then so is $(P,  \bar{P}', \bar{Q}')$. Similarly, the objective function is invariant under the same transformation by \Cref{lem:ardp-symmetric}. Now, consider averaging our optimal $(P, P_1', Q_1')$ and $(\bar{P},  \bar{P}', \bar{Q}')$ to generate $(P, \tilde{P'}, \tilde{Q'})$. Due to convexity, this will also be optimal. Furthermore,

\begin{align*}
\tilde{P'}(x) &= \frac{1}{2}(P_1'(x) + \bar{P'}(x)) = \frac{1}{2}(P_1'(x) + Q_1'(n_d - x))\\
\tilde{Q'}(x) &= \frac{1}{2}(Q_1'(x) + \bar{Q'}(x)) =\frac{1}{2}(Q_1'(x) + P_1'(n_d - x)).
\end{align*}

Therefore $\tilde{Q'}(n_d - x) = \frac{1}{2}(Q_1'(n_d - x) + P_1'(n_d - (n_d - x))) = \tilde{P'}(x)$.

This allows us to further eliminate $Q_1$ from \Cref{eq:opt-additive} by replacing it with a shifted and reflected version of $P_1$.
\end{proof}

While the optimization problem in \Cref{thm:opt-additive-simplified} is convex, its objective is not numerically stable for large values of $\alpha$. In practice, we perform the optimization in log-probabilities (which is not convex), and certify optimality by measuring the duality gap (using complementary slackness) of the convex objective in linear space. \Cref{fig:additive-distributions} plots optimal distributions at various values of $\alpha$. The plots suggest that as $\alpha \to \infty$, the optimal distribution for $\pi(n_d) = 1$ converges to truncated discrete Laplace. We prove this convergence below, which aligns with the findings in \cite{desfontaines-partition}.

Furthermore, we plot the \emph{privacy} of the optimal additive mechanism (for a specific point-wise utility guarantee) against $\pi^*$ in \Cref{fig:additive-comp}, showing a numerical separation in privacy even in this weaker setting where the additive noise needs to compete with $\pi^*$ only at a single point $n_d$.

\begin{proposition}\label{prop:converge-discrete-laplace}
For a fixed odd $n_d > 2$ and $\delta \in (0, 1/n_d]$, the unique distribution $X$ supported on $\{0, \dots, n_d - 1\}$ that minimizes $\eps$ subject to $(\eps, \delta)$-DP for $\pi_X(n_d) = 1$ is a truncated discrete Laplace distribution:
$$
P(x) = A \cdot e^{-\eps |x - \mu|}
$$
where $\mu = \frac{n_d - 1}{2}$ and $A$ is the constant such that $\sum P(x) = 1$. 

Furthermore, when $\delta \in (1/n_d, 1]$, the minimum $\eps$ is 0 which can be achieved when $X$ is the uniform distribution. (However, this is not the unique distribution achieving $(0, \delta)$-DP.)
\end{proposition}

In the following proof, we use $[x]_+$ as a shorthand for $\max\{x, 0\}$. Furthermore, recall that the $(\eps, \delta)$-DP can be characterized by the hockey stick divergence, which will be more convenient for us to work with. Namely, let
\begin{align*}
\hs{P}{Q}{\eps} = \int_{\Omega} \Brac{P(x) - e^{\eps} \cdot Q(x)}_+ dx.
\end{align*}
It is known that $D^{\delta}_{\infty}(P \| Q) \leq \eps$ if and only if $\hs{P}{Q}{\eps} \leq \delta$. We will use the latter formulation in the proof below.

\begin{proof}
We will first consider the case $\delta \in (0, 1/n_d]$. 
Define $Z(\eps) = \delta  \sum_{x=0}^{n_d - 1} e^{\eps (\mu - |x - \mu|)}$. Let $\eps^*$ be the unique\footnote{The uniqueness follows from monotonicity of $Z(\eps)$, $Z(0) = n_d \delta \leq 1$ and $\lim_{\eps \to \infty} Z(\eps) = \infty$.} non-negative real number such that $Z(\eps^*) = 1$. Finally, let $P^*$ denote the truncated discrete Laplace distribution supported on $\{0, \dots, n_d - 1\}$ with $P^*(x) = \delta \cdot e^{\eps^* (\mu - |x - \mu|)}$. This is a valid distribution since $Z(\eps^*) = 1$, and it is simple to check that $\hs{P^*}{P^*_+}{\eps^*} = \delta$, implying that the resulting additive mechanism is $(\eps^*, \delta)$-DP as desired.

Suppose for the sake of contradiction that there is some other distribution $\tP$ supported on $\{0, \dots, n_d - 1\}$ that satisfies $(\teps, \delta)$-DP for some $\teps < \eps$. By \Cref{prop:opt-symmetry2}, we may assume without loss of generality that $\tP$ is symmetric around $\mu$. Then, consider any $i \in \{1, \dots, n_d - 1\}$. Since $\tP$ satisfies $(\teps, \delta)$-DP, we have
\begin{align*}
\delta \geq \hs{\tP}{\tP_+}{\teps} &= \sum_{x \in \{0, \dots, n_d - 1\}} \Brac{\tP(x) - e^{\teps} \cdot \tP(x - 1)}_+  \\
&\geq \sum_{x \in \{0, \dots, i\}} e^{-\teps x} \cdot \Brac{\tP(x) - e^{\teps} \cdot \tP(x - 1)}_+  \\
&\geq \sum_{x \in \{0, \dots, i\}} e^{-\teps x} \cdot \Paren{\tP(x) - e^{\teps} \cdot \tP(x - 1)}  \\
&= e^{-\teps x} \cdot \tP(x). 
\end{align*}
From this and from symmetry of $\tP$, we have
\begin{align} \label{eq:point-wise-bound-dlap}
\tP(x) \leq \delta \cdot \min\{e^{\teps x}, e^{\teps (n_d - 1 - x)}\} = \delta \cdot e^{\teps \left(\mu - |x - \mu|\right)} < P^*(x).
\end{align}

However, this means that
\begin{align*}
\sum_{x \in \{0, \dots, n_d - 1\}} \tP(x) <  \sum_{x \in \{0, \dots, n_d - 1\}} P^*(x) = 1,
\end{align*}
which is a contradiction since $\tP$ is a distribution. 

To see the uniqueness of $P^*$ for $(\eps^*, \delta)$-DP, notice that, when $\teps = \eps^*$, \eqref{eq:point-wise-bound-dlap} becomes $\tP(x) \leq P^*(x)$ for all $x$. Since both $\tP$ and $P^*$ are distributions, this implies that they must be the same.

For the case of $\delta \in (1/n_d, 1]$, it is simple to verify that the uniform distribution satisfies $(\eps, \delta)$-DP for $\eps = 0$, and this is the smallest $\eps$ possible.
\end{proof}

We end this section by remarking that noise with bounded support has applications in DP beyond partition selection (e.g. in multi-party DP~\cite{BGGKMRS22}). Thus, the tools developed here might be useful elsewhere; we leave this to future work.
\section{Conclusion}

This work continues the direction started by \cite{desfontaines-partition} in deriving \emph{optimal} mechanisms for partition selection. We generalize their result to approximate RDP, handling finite values of $\alpha$. When users only submit a single partition, this result alone is of practical utility in cases where the mechanism will be composed many times due to the tighter composition properties of (approximate) RDP.

In the finite $\alpha$ regime we also show that there is a separation in privacy between mechanisms which are post-processing of additive noise mechanisms and those which are not; there is an inherent cost to ``releasing the weight" of a partition. This differs from the $\alpha \to \infty$ regime as noted by \cite{desfontaines-partition}, and implies that if the partition weight is not needed, analysts and mechanism designers should strongly consider using non-additive mechanisms to maximize utility.

Furthermore, we use our recipe for the optimal mechanism to create a selection mechanism targeting $L^r$ norm constraints on vector contributions. When $r = 1$ it gives us a drop-in replacement for the Laplace mechanism which is used as a subroutine in \cite{Carvalho2022IncorporatingIF}. When $r=2$ it gives us a drop-in replacement for the Gaussian mechanism which is used as a subroutine in \cite{swanberg2023dp}, \cite{gopi20a} and \cite{chen2025scalable}. When injecting SNAPS into the latter two algorithms, we show improved utility across the board over a variety of datasets.

There are several interesting open questions left by our work, as highlighted below.
\begin{itemize}
\item \textbf{Beyond Drop-In Replacement:} Our focus has been to devise partition selection primitives that can be used as drop-in replacements for additive noise mechanisms, such as the (truncated) Gaussian and Laplace. However, some sophisticated partition selection mechanisms--including those used in our experiments~\cite{chen2025scalable,gopi20a}--operate in multiple stages, which are more tailored towards $L^r$ sensitivity. It would be interesting to see if these intermediate steps can be designed in conjunction with the optimal partition selection primitive for better compatability and overall utility.
\item \textbf{Tighter Composition:} Our primary motivation for considering (approximate) RDP is due to its composition property. Another technique for tight privacy accouting which is widely used is through the so-called privacy loss distributions (PLDs)~\cite{meiser2018tight,sommer2019privacy,koskela2020computing,gopi21numerical,ghazi22faster,doroshenko22connect}. Using PLD requires tracking the \emph{dominating pairs} (aka ``worst case'') of neighboring output distributions~\cite{zhu22optimal}. For our mechanisms, there do not seem to be a tight dominating pair. Thus, it is an interesting direction to design a partition selection mechanism that is more amenable to PLD-based privacy accounting.
\end{itemize}

\bibliographystyle{alpha}
\bibliography{ref}

@inproceedings{renyi61,
    author = {Alfr{\'e}d R{\'e}nyi},
    title = {On measures of entropy and information},
    booktitle = {Proceedings of the Fourth Berkeley Symposium on Mathematical
  Statistics and Probability, Volume 1: Contributions to the Theory of
  Statistics},
    year = {1961}
}

@inproceedings{Mironov17,
  author       = {Ilya Mironov},
  title        = {R{\'{e}}nyi Differential Privacy},
  booktitle    = {CSF},
  pages        = {263--275},
  year         = {2017},
  url          = {https://doi.org/10.1109/CSF.2017.11},
  doi          = {10.1109/CSF.2017.11},
  bibsource    = {dblp computer science bibliography, https://dblp.org}
}

@inproceedings{dwork-calibrating,
author = {Dwork, Cynthia and McSherry, Frank and Nissim, Kobbi and Smith, Adam},
title = {Calibrating noise to sensitivity in private data analysis},
year = {2006},
isbn = {3540327312},
url = {https://doi.org/10.1007/11681878_14},
doi = {10.1007/11681878_14},
booktitle = {TCC},
pages = {265–284},
numpages = {20},
}

@article{desfontaines-partition,
  title = {Differentially private partition selection},
  volume = {2022},
  ISSN = {2299-0984},
  url = {http://dx.doi.org/10.2478/popets-2022-0017},
  DOI = {10.2478/popets-2022-0017},
  number = {1},
  journal = {PoPETS},
  author = {Desfontaines,  Damien and Voss,  James and Gipson,  Bryant and Mandayam,  Chinmoy},
  year = {2021},
  month = nov,
  pages = {339–352}
}

@article{swanberg2023dp,
  title={DP-SIPS: A simpler, more scalable mechanism for differentially private partition selection},
  author={Swanberg, Marika and Desfontaines, Damien and Haney, Samuel},
  journal={PoPETS},
  year={2023}
}

@inproceedings{
chen2025scalable,
title={Scalable Private Partition Selection via Adaptive Weighting},
author={Justin Y. Chen and Vincent Cohen-Addad and Alessandro Epasto and Morteza Zadimoghaddam},
booktitle={ICML},
year={2025},
url={https://openreview.net/forum?id=NEfIaE4BI5}
}

@article{gilani2025optimal,
  title={Optimal Additive Noise Mechanisms for Differential Privacy},
  author={Gilani, Atefeh and Gomez, Juan Felipe and Asoodeh, Shahab and Calmon, Flavio P and Kosut, Oliver and Sankar, Lalitha},
  journal={arXiv preprint arXiv:2504.14730},
  year={2025}
}

@inproceedings{dworkcalibrating,
author = {Dwork, Cynthia and McSherry, Frank and Nissim, Kobbi and Smith, Adam},
title = {Calibrating noise to sensitivity in private data analysis},
year = {2006},
isbn = {3540327312},
url = {https://doi.org/10.1007/11681878_14},
doi = {10.1007/11681878_14},
booktitle = {TCC},
pages = {265–284},
numpages = {20},
}

@article{plume2022,
  publtype={informal},
  author={Kareem Amin and Jennifer Gillenwater and Matthew Joseph and Alex Kulesza and Sergei Vassilvitskii},
  title={Plume: Differential Privacy at Scale},
  year={2022},
  cdate={1640995200000},
  journal={CoRR},
  volume={abs/2201.11603},
  url={https://arxiv.org/abs/2201.11603}
}

@article{wilson2020differentially,
  title={Differentially Private SQL with Bounded User Contribution},
  author={Wilson, Royce J and Zhang, Celia Yuxin and Lam, William and Desfontaines, Damien and Simmons-Marengo, Daniel and Gipson, Bryant},
  journal={PoPETS},
  volume={2},
  pages={230--250},
  year={2020}
}

@InProceedings{gopi20a,
  title = 	 {Differentially Private Set Union},
  author =       {Gopi, Sivakanth and Gulhane, Pankaj and Kulkarni, Janardhan and Shen, Judy Hanwen and Shokouhi, Milad and Yekhanin, Sergey},
  booktitle = 	 {ICML},
  pages = 	 {3627--3636},
  year = 	 {2020},
  url = 	 {https://proceedings.mlr.press/v119/gopi20a.html},
}

@inproceedings{Korolova-search-queries,
author = {Korolova, Aleksandra and Kenthapadi, Krishnaram and Mishra, Nina and Ntoulas, Alexandros},
title = {Releasing search queries and clicks privately},
year = {2009},
isbn = {9781605584874},
url = {https://doi.org/10.1145/1526709.1526733},
doi = {10.1145/1526709.1526733},
booktitle = {WWW},
pages = {171–180},
}

@inproceedings{
papernot2022hyperparameter,
title={Hyperparameter Tuning with Renyi Differential Privacy},
author={Nicolas Papernot and Thomas Steinke},
booktitle={ICLR},
year={2022},
url={https://openreview.net/forum?id=-70L8lpp9DF}
}

@article{van2014renyi,
  title={R{\'e}nyi divergence and Kullback-Leibler divergence},
  author={Van Erven, Tim and Harremos, Peter},
  journal={IEEE Transactions on Information Theory},
  volume={60},
  number={7},
  pages={3797--3820},
  year={2014},
  publisher={IEEE}
}

@misc{DPLib,
  author = {Differential Privacy Team, Google},
  title = {Privacy Loss Distributions},
  year         = {2026},
  url = {https://github.com/google/differential-privacy/blob/main/},
  urldate = {2026-01-24}
}

@article{dwork-roth,
author = {Dwork, Cynthia and Roth, Aaron},
title = {The Algorithmic Foundations of Differential Privacy},
year = {2014},
issue_date = {Aug 2014},
publisher = {Now Publishers Inc.},
address = {Hanover, MA, USA},
volume = {9},
number = {3–4},
issn = {1551-305X},
url = {https://doi.org/10.1561/0400000042},
doi = {10.1561/0400000042},
journal = {Found. Trends Theor. Comput. Sci.},
month = aug,
pages = {211–407},
numpages = {197}
}

@article{canonne2020discrete,
  title={The discrete gaussian for differential privacy},
  author={Canonne, Cl{\'e}ment L and Kamath, Gautam and Steinke, Thomas},
  journal={NeurIPS},
  volume={33},
  pages={15676--15688},
  year={2020}
}

@inproceedings{Carvalho2022IncorporatingIF,
  title={Incorporating Item Frequency for Differentially Private Set Union},
  author={Ricardo Silva Carvalho and Ke Wang and Lovedeep Gondara},
  booktitle={AAAI},
  year={2022},
  url={https://api.semanticscholar.org/CorpusID:250289456}
}

@inproceedings{BGGKMRS22,
  author       = {James Bell and
                  Adri{\`{a}} Gasc{\'{o}}n and
                  Badih Ghazi and
                  Ravi Kumar and
                  Pasin Manurangsi and
                  Mariana Raykova and
                  Phillipp Schoppmann},
  title        = {Distributed, Private, Sparse Histograms in the Two-Server Model},
  booktitle    = {CCS},
  pages        = {307--321},
  year         = {2022},
  url          = {https://doi.org/10.1145/3548606.3559383},
  doi          = {10.1145/3548606.3559383},
  bibsource    = {dblp computer science bibliography, https://dblp.org}
}

@inproceedings{zhu22optimal,
  author       = {Yuqing Zhu and
                  Jinshuo Dong and
                  Yu{-}Xiang Wang},
  title        = {Optimal Accounting of Differential Privacy via Characteristic Function},
  booktitle    = {AISTATS},
  pages        = {4782--4817},
  year         = {2022},
}

@inproceedings{meiser2018tight,
  title={Tight on budget? {T}ight bounds for $r$-fold approximate differential privacy},
  author={Meiser, Sebastian and Mohammadi, Esfandiar},
  booktitle={CCS},
  pages={247--264},
  year={2018}
}

@article{sommer2019privacy,
  author       = {David M. Sommer and
                  Sebastian Meiser and
                  Esfandiar Mohammadi},
  title        = {Privacy Loss Classes: The Central Limit Theorem in Differential Privacy},
  journal      = {PoPETS},
  volume       = {2019},
  number       = {2},
  pages        = {245--269},
  year         = {2019},
}

@inproceedings{koskela2020computing,
  title={Computing tight differential privacy guarantees using {FFT}},
  author={Koskela, Antti and J{\"a}lk{\"o}, Joonas and Honkela, Antti},
  booktitle={AISTATS},
  pages={2560--2569},
  year={2020},
}

@inproceedings{gopi21numerical,
  author    = {Sivakanth Gopi and
               Yin Tat Lee and
               Lukas Wutschitz},
  title     = {Numerical Composition of Differential Privacy},
  booktitle = {NeurIPS},
  pages     = {11631--11642},
  year      = {2021},
}

@article{doroshenko22connect,
  author       = {Vadym Doroshenko and
                  Badih Ghazi and
                  Pritish Kamath and
                  Ravi Kumar and
                  Pasin Manurangsi},
  title        = {Connect the Dots: Tighter Discrete Approximations of Privacy Loss
                  Distributions},
  journal      = {PoPETS},
  volume       = {2022},
  number       = {4},
  pages        = {552--570},
  year         = {2022},
}

@inproceedings{ghazi22faster,
  author    = {Badih Ghazi and
               Pritish Kamath and
               Ravi Kumar and
               Pasin Manurangsi},
  title     = {Faster Privacy Accounting via Evolving Discretization},
  booktitle = {ICML},
  pages     = {7470--7483},
  year      = {2022},
}

@misc{twitter,
	title={Customer Support on Twitter},
	url={https://www.kaggle.com/dsv/8841},
	DOI={10.34740/KAGGLE/DSV/8841},
	publisher={Kaggle},
	author={Stuart Axelbrooke},
	year={2017}
}

@misc{finance,
	title={Daily Financial News for 6000+ Stocks},
	url={https://www.kaggle.com/miguelaenlle/datasets},
	publisher={Kaggle},
	author={Miguel Aenlle}
}

@misc{wiki,
	title={Simple/Normal Wikipedia Abstracts V1},
	url={https://www.kaggle.com/datasets/markwijkhuizen/simplenormal-wikipedia-abstracts-v1},
	publisher={Kaggle},
	author={Mark Wijkhuizen}
}

@inproceedings{amazon2,
 author = {Zhang, Xiang and Zhao, Junbo and LeCun, Yann},
 booktitle = {Advances in Neural Information Processing Systems},
 editor = {C. Cortes and N. Lawrence and D. Lee and M. Sugiyama and R. Garnett},
 pages = {},
 publisher = {Curran Associates, Inc.},
 title = {Character-level Convolutional Networks for Text Classification},
 url = {https://proceedings.neurips.cc/paper_files/paper/2015/file/250cf8b51c773f3f8dc8b4be867a9a02-Paper.pdf},
 volume = {28},
 year = {2015}
}

@article{amazon1,
  title={Hidden factors and hidden topics: understanding rating dimensions with review text},
  author={Julian McAuley and Jure Leskovec},
  journal={Proceedings of the 7th ACM conference on Recommender systems},
  year={2013},
  url={https://api.semanticscholar.org/CorpusID:6440341}
}

@InProceedings{imdb,
  author    = {Maas, Andrew L.  and  Daly, Raymond E.  and  Pham, Peter T.  and  Huang, Dan  and  Ng, Andrew Y.  and  Potts, Christopher},
  title     = {Learning Word Vectors for Sentiment Analysis},
  booktitle = {Proceedings of the 49th Annual Meeting of the Association for Computational Linguistics: Human Language Technologies},
  month     = {June},
  year      = {2011},
  address   = {Portland, Oregon, USA},
  publisher = {Association for Computational Linguistics},
  pages     = {142--150},
  url       = {http://www.aclweb.org/anthology/P11-1015}
}

\appendix

\section{Water-Filling Algorithm}
\label{app:waterfilling}

We give more detail of the ``water-filling'' algorithm from \Cref{sec:additive-rdp}. Throughout this section, we assume that the total variation distance of $A, B$ is greater than $\delta$ (i.e. $\sum_x \min\{A(x), B(x)\} < 1 - \delta$). Otherwise, the approximate R{\'e}nyi divergence is simply zero.

\subsection{$O(n \log n)$ Implementation of the Algorithm}

\noindent To implement the algorithm efficiently, we observe that computing the cutoffs $\lambda_P$ and $\lambda_Q$ requires finding a threshold that truncates exactly $\delta$ probability mass from the tail of a likelihood ratio distribution. We define a generic subprocedure $\textsc{ComputeCutoff}$ to compute this threshold. This subprocedure takes as input probability distributions $A, B$ over a discrete domain $\mathcal{X}$ of size $n$, and a target mass $\delta \in (0,1)$. It then outputs a ratio $\lambda > 1$ such that $\sum_x \min\{A(x), \lambda \cdot B(x)\} = 1 - \delta$.

The subprocedure $\textsc{ComputeCutoff}(A, B, \delta)$ works as follows:
\begin{enumerate}
    \item \textbf{Sort Ratios:} Compute the likelihood ratios $r(x) = \frac{A(x)}{B(x)}$ for all $x \in \mathcal{X}$. Sort the domain to $x_1, x_2, \dots, x_n$ such that $r(x_1) \ge r(x_2) \ge \dots \ge r(x_n)$. Define an auxiliary bound $r(x_{n+1}) = 0$.
    
    \item \textbf{Initialize:} Set cumulative masses $S_A = 0$ and $S_B = 0$.
    
    \item \textbf{Linear Scan:} For $k = 1, 2, \dots, n$:
    \begin{enumerate}
        \item Update prefix sums: $S_A \gets S_A + A(x_k)$ and $S_B \gets S_B + B(x_k)$.
        \item Calculate the hypothetical mass removed if the cutoff were exactly $r(x_{k+1})$:
        \begin{equation*}
            \delta_k = S_A - r(x_{k+1}) S_B
        \end{equation*}
        \item \textbf{Termination:} If $\delta_k \ge \delta$, the exact cutoff $\lambda$ lies within $[r(x_{k+1}), r(x_k)]$. Because the removed mass is strictly linear with respect to the cutoff in this interval, we simply return:
        \begin{equation*}
            \lambda = \frac{S_A - \delta}{S_B}
        \end{equation*}
    \end{enumerate}
\end{enumerate}

\paragraph{Main Algorithm.}
The cutoffs in our algorithm can then be computed via two independent calls to the subprocedure: \begin{align*}
    \lambda_P = \textsc{ComputeCutoff}(P, Q, \delta), & &\lambda_Q = \frac{1}{\textsc{ComputeCutoff}(Q, \tP, \delta)}.
\end{align*}
This results in $O(n \log n)$-time algorithm as claimed.

\subsection{Proof of Correctness}

\noindent \textbf{Problem Formulation.} 
Recall that we define $\tP(x) = (1 - \delta)P'(x)$ and $\tQ(x) = (1 - \delta)Q'(x)$.
The equivalent minimization problem is:
\begin{align*}
    \text{minimize } & \sum_{x} \tilde{P}(x)^\alpha \tilde{Q}(x)^{1-\alpha} \\
    \text{subject to } & \sum_{x} \tilde{P}(x) = 1-\delta, \quad \sum_{x} \tilde{Q}(x) = 1-\delta, \\
                       & \tilde{P}(x) - P(x) \le 0, \quad \tilde{Q}(x) - Q(x) \le 0, \\
                       & -\tilde{P}(x) \le 0, \quad -\tilde{Q}(x) \le 0.
\end{align*}
Note that, for $\alpha > 1$, the objective function $f(\tilde{P}, \tilde{Q}) = \sum_x \tilde{P}(x)^\alpha \tilde{Q}(x)^{1-\alpha}$ is strictly convex.

\paragraph{Lagrangian.}
Let $\nu_P, \nu_Q \in \mathbb{R}$ be the multipliers for the equality constraints, and $\mu_P(x), \mu_Q(x) \in \mathbb{R}$ be the multipliers for the inequality constraints. The Lagrangian is\footnote{We only consider solutions which are interior points for the lower bounds (non-negativity strictly satisfied) and we thus omit their multipliers. Note that this is without loss of generality since our algorithm's solution satisfies this for all $x$ such that $P(x), Q(x) > 0$. (For all other $x$'s, we are forced to set $\tilde{P}(x) = \tilde{Q}(x) = 0$ anyway.)}:
\begin{align*}
    \mathcal{L}(\tilde{P}, \tilde{Q}, \nu_P, \nu_Q, \mu_P, \mu_Q) 
    &= \sum_{x} \tilde{P}(x)^\alpha \tilde{Q}(x)^{1-\alpha} \\
    &\quad - \nu_P \left( \sum_{x} \tilde{P}(x) - (1-\delta) \right) - \nu_Q \left( \sum_{x} \tilde{Q}(x) - (1-\delta) \right) \\
    &\quad + \sum_{x} \mu_P(x) (\tilde{P}(x) - P(x)) + \sum_{x} \mu_Q(x) (\tilde{Q}(x) - Q(x)).
\end{align*}

\paragraph{Output Solution and Lagrangian Multipliers.}
Let $\tilde{P}$ and $\tilde{Q}$ be the solutions produced by the water-filling algorithm. 
We note that $\lambda_P \geq 1 \geq \lambda_Q$. The former inequality follows from our assumption that $\sum_x \min\{P(x), Q(x)\} < 1 - \delta$, while the latter follows from $1 - \delta = \sum_x \tQ(x) \leq \sum_x \tP(x) / \lambda_Q = (1 - \delta) / \lambda_Q$.

Let $R(x) = \frac{\tilde{P}(x)}{\tilde{Q}(x)}$.
The sample space can be partitioned into three regions based on the ratio $\frac{P(x)}{Q(x)}$: 
\begin{itemize}
    \item \textbf{Region 1} ($P(x)/Q(x) > \lambda_P$): We have $\tP(x) = \lambda_P Q(x)$, $\tQ(x) = Q(x) \implies R(x) = \lambda_P$.
    \item \textbf{Region 2} ($\lambda_Q \le P(x)/Q(x) \le \lambda_P$): We have $\tP(x) = P(x)$, $\tQ(x) = Q(x) \implies R(x) = \frac{P(x)}{Q(x)}$.
    \item \textbf{Region 3} ($P(x)/Q(x) < \lambda_Q$): We have $\tP(x) = P(x)$, $\tQ(x) = \frac{P(x)}{\lambda_Q} \implies R(x) = \lambda_Q$.
\end{itemize}
We explicitly define the Lagrangian multipliers for this proposed solution as follows:
\begin{align*}
    \nu_P &= \alpha \lambda_P^{\alpha-1} \\
    \nu_Q &= (1-\alpha) \lambda_Q^\alpha \\
    \mu_P(x) &= \nu_P - \alpha R(x)^{\alpha-1} = \alpha \left( \lambda_P^{\alpha-1} - R(x)^{\alpha-1} \right) \\
    \mu_Q(x) &= \nu_Q - (1-\alpha) R(x)^\alpha = (\alpha - 1) \left( R(x)^\alpha - \lambda_Q^\alpha \right)
\end{align*}

\paragraph{Verification of KKT Conditions.}

Since the problem is strictly convex, the KKT conditions are sufficient for global optimality. We verify that our proposed solution and multipliers satisfy all four KKT conditions:

\begin{enumerate}
    \item \textbf{Stationarity:}
    By taking the gradient of $\mathcal{L}$ and substituting our defined multipliers, stationarity holds by construction for all $x$:
    \begin{align*}
        \frac{\partial \mathcal{L}}{\partial \tilde{P}(x)} &= \alpha R(x)^{\alpha-1} - \nu_P + \mu_P(x) = 0 \\
        \frac{\partial \mathcal{L}}{\partial \tilde{Q}(x)} &= (1-\alpha) R(x)^\alpha - \nu_Q + \mu_Q(x) = 0
    \end{align*}

    \item \textbf{Primal Feasibility:}
    By the definition of the algorithm, $0 \le \tilde{P}(x) \le P(x)$ and $0 \le \tilde{Q}(x) \le Q(x)$. The chosen cutoffs $\lambda_P, \lambda_Q$ explicitly guarantee $\sum \tilde{P}(x) = 1-\delta$ and $\sum \tilde{Q}(x) = 1-\delta$.

    \item \textbf{Dual Feasibility:}
    As discussed above, $R(x) \in [\lambda_P, \lambda_Q]$. Since $\alpha > 1$, we thus have $\mu_P(x) = \alpha(\lambda_P^{\alpha-1} - R(x)^{\alpha-1}) \ge 0$ and $\mu_Q(x) = (\alpha-1)(R(x)^\alpha - \lambda_Q^\alpha) \ge 0$.
    \item \textbf{Complementary Slackness:} 
    \begin{itemize}
        \item For $\tilde{P}$: If $\tilde{P}(x) < P(x)$, then $x$ must be in Region 1, meaning $R(x) = \lambda_P$. Substituting this yields $\mu_P(x) = \alpha(\lambda_P^{\alpha-1} - R(x)^{\alpha-1}) = 0$. Thus, $\mu_P(x) (\tilde{P}(x) - P(x)) = 0$ for all $x$.
        \item For $\tilde{Q}$: If $\tilde{Q}(x) < Q(x)$, then $x$ must be in Region 3, meaning $R(x) = \lambda_Q$. Substituting this yields $\mu_Q(x) = (\alpha-1)(\lambda_Q^\alpha - \lambda_Q^\alpha) = 0$. Thus, $\mu_Q(x) (\tilde{Q}(x) - Q(x)) = 0$ for all $x$.
    \end{itemize}
\end{enumerate}

The solution generated by the sequential water-filling algorithm, coupled with the defined multipliers, satisfies  the KKT sufficiency conditions. As a result, the solution is optimal. $\hfill \qedsymbol$

\end{document}